\documentclass[12pt,a4paper]{article}
\usepackage[english]{babel}
\usepackage{amsmath,amssymb,amsthm,epsfig,color,graphicx, mathtools, bbm, dsfont, geometry, float, mathrsfs, upgreek, bm, enumitem, comment, hyperref}
\usepackage[latin1]{inputenc}
\usepackage[T1]{fontenc}
\usepackage{indentfirst}
\usepackage[title]{appendix}
\usepackage{addfont}
\addfont{OT1}{rsfs10}{\rsfs}

\usepackage{tikz, tkz-euclide}
\usetikzlibrary{calc}
\usetikzlibrary{decorations.markings}
\usepackage{fancybox}

\usepackage{graphicx}
\graphicspath{ {images/} }
\newcommand{\Z}{\mathbb{Z}}

\newcommand{\diam}{\mathrm{diam}}

\newcommand{\be}{\begin{equation}}
	\newcommand{\ee}{\end{equation}}

\numberwithin{equation}{section}

\usepackage{bbm}
\newtheorem*{theorem*}        {Theorem}
\newtheorem*{conjecture*}   {Conjecture}
\newtheorem{theorem}           {Theorem}[section]
\newtheorem{lemma}              {Lemma}[section]
\newtheorem*{lemma*}          {Lemma}

\newtheorem{definition}         {Definition}[section]
\newtheorem{example}         {Example}[section]
\newtheorem{corollary}          {Corollary}[section]
\newtheorem{proposition}      {Proposition}[section]

\makeatletter

\newcommand\longleftrightarrowfill@{%
	\arrowfill@\leftarrow\relbar\rightarrow}
\makeatother

\definecolor{Red}{cmyk}{0,1,1,0}

\definecolor{Blue}{cmyk}{1,1,0,0}

\definecolor{DarkBlue}{rgb}{0.1,0.1,0.5}
\definecolor{Red}{rgb}{0.9,0.0,0.1}
\definecolor{DarkGreen}{rgb}{0.10,0.50,0.10}
\definecolor{DarkRed}{rgb}{0.50,0.10,0.10}
\definecolor{bleu}{RGB}{0,140,189}%

\newgeometry{vmargin={15mm}, hmargin={12mm,17mm}}

\begin{document}

	\begin{center}
		{\LARGE Quantum Statistical Mechanics via Boundary Conditions.  \\A Groupoid Approach to Quantum Spin Systems}
		\vskip.5cm
		Lucas Affonso$^{1}$, Rodrigo Bissacot$^{1,2}$, Marcelo Laca$^{3}$
		\vskip.3cm
		\begin{footnotesize}
			$^{1}$Institute of Mathematics and Statistics (IME-USP), University of S\~{a}o Paulo, Brazil\\
			$^{2}$ Faculty of Mathematics and Computer Science, Nicolaus Copernicus University, Poland\\  
			$^{3}$Department of Mathematics and Statistics, University of Victoria, Victoria, BC V8W 3R4, Canada
		\end{footnotesize}
		\vskip.1cm
		\begin{scriptsize}
			emails: lucas.affonso.pereira@gmail.com, rodrigo.bissacot@gmail.com, laca@uvic.ca
		\end{scriptsize}
		
	\end{center}
	
	\begin{abstract}
		We use a groupoid model for the spin algebra to introduce boundary conditions on quantum spin systems via a Poisson point process representation. We can describe KMS states of quantum systems by means of a set of equations resembling the standard DLR equations of classical statistical mechanics. We introduce a notion of quantum specification which recovers the classical DLR measures in the particular case of classical interactions. Our results are in the same direction as those obtained recently by Cha, Naaijkens, and Nachtergaele, differently somehow from the predicted by Fannes and Werner.
	\end{abstract}
	
	\section{Introduction}
	
The {\it KMS condition} was proposed as a characterization for equilibrium states in quantum statistical mechanics in a seminal paper by Haag, Hugenholtz, and Winnink \cite{HHW}. Around the same time, Ruelle-Lanford \cite{Lan} and Dobrushin \cite{Dob1} proposed a set of equations, known nowadays as the {\it DLR equations}, that characterize the equilibrium states for classical statistical mechanics systems. 
	
	Both conditions are, at first sight, completely different. The KMS condition is a dynamical condition for equilibrium, relying on well-defined dynamics on the observable algebra to make sense, while the DLR equations characterize the conditional expectations of the state concerning the $\sigma-$algebra of the events localized outside of a finite region of $\Z^d$. Also, some theorems pose problems with a direct analogy between both conditions. 
	
		The first evidence is given by Proposition 5.3.28 in \cite{Bra2}, which implies that the dynamics on a classical algebra of observables must be trivial. Another theorem is Takesaki's theorem \cite{Acc1}, which states that if there is a conditional expectation between the quasi-local algebra of observables and the localized algebra in a finite region $\Lambda \subset \Z^d$, then the state factors as a product state. 
	
At this point, we can ask about the relation between the two characterizations. Brascamp \cite{Bras} showed that a state is KMS for classical interactions of lattice gases when embedded in the CAR algebra, if and only if they are DLR probability measures on the configuration space. For general finite spin systems, this was solved in two papers by Araki and Ion \cite{Ara1}, which solved the one-dimensional case and the high-temperature case for all dimensions, and by another article by Araki, in \cite{Ara2}, where the result was proven in full generality. 
	
	To this end, Araki introduced the so-called \emph{Gibbs condition}, which is equivalent to the KMS condition (Theorem 6.2.18 from \cite{Bra2}) and reduces to the DLR equations if the interaction is classical. This condition relies on perturbation theory with bounded operators developed by Araki, which poses some difficulties that are absent in the classical case when one tries to identify the pure phases of the system since the definition of the Gibbs condition involves modular automorphisms. In the words of T. Matsui, in \cite{Matsui},
	
	\begin{quotation}
		"One mathematically interesting question is whether any KMS state is obtained in this procedure, namely, one may ask whether any KMS state is a thermodynamic limit of finite volume Gibbs states with suitable boundary conditions for Hamiltonians as is described here. Theorem 3.3 may be taken as an answer to this question, however, this is not what we want. We are asking the effect of \emph{the boundary condition of our Hamiltonian} in a large system while the Gibbs condition is \emph{the boundary condition imposed on the states}. From a practical point of view, the [Gibbs condition] is cumbersome to handle [\dots]. We are not certain that all the infinite volume Gibbs states are obtained via a [thermodynamic limit]."
	\end{quotation}
	
Indeed, from a classical statistical mechanics point of view, all the pure phases of the system are obtained via a thermodynamic limit (see Theorem 7.12 in \cite{Geo}). However, it is known that there exist examples of non-extremal DLR measures which can not be obtained via a thermodynamic limit, see \cite{Co, Mi}. In the quantum case, there are proposals by  Israel \cite{Is} and  Simon \cite{Simon}. For Israel, the boundary conditions should consist of conditioning the Hamiltonian to a fixed pure state of the C*-algebra outside the box. In Simon's approach, the boundary conditions come from states that can be described by a family of compatible density matrices. Of course, if the interaction in consideration has a finite range, both proposals are the same. 
	
	Other approaches for studying states in quantum spin systems with boundary conditions were made using Poisson point processes; see \cite{Aiz1, Aiz2}. In their papers, M. Aizenman and B. Nachtergaele proposed the notion of quasi-states, a linear function that has a positive restriction to an abelian subalgebra. Although promising, no relation with KMS states has been further investigated since then. In \cite{Wer}, M. Fannes and R. F. Werner raised concerns about the validity of what they called a DLR inclusion, even suggesting some counterexamples to constructing such a theory for quantum systems. But in recent years, M. Cha, P. Naaijkens, and B. Nachtergaele \cite{Pieter} have characterized the ground states of Kitaev models using a suitable notion of boundary operators, reigniting our interest and showing that maybe a DLR theory could be developed. By analyzing the monograph by Gruber, Hintermann, and Merlini \cite{Gruber}, we noticed that the transformation group studied there had the usual spin algebra as its $\text{C}^*$-algebra, giving us a possible "quantum space" to act as a configuration space for a DLR theory. Thus by combining random representations and the groupoid model for the $\text{C}^*$-algebra of the quantum spin system, we could find a suitable generalization of the DLR equations for the quantum setting. In this paper, we show some basic properties that a quantum specification should have and show that the states compatible with it are KMS states. In fact, DLR and conformal measures already were considered in the setting of groupoid $\text{C}^*$-algebra algebras exploring the connections with the results in thermodynamic formalism on countable Markov shifts, see \cite{BEFR, BEFR2}. 
 
Our paper is divided as follows. In Section 2, we give some basic definitions, state some known-results on groupoid and their C*-algebras, and about Poisson point processes representations applied to short-range models in quantum spin systems. 
In section 3, we introduce the notions finite volume quantum Gibbs states depending on paths. We also introduce a definition of quantum specification for groupoids and show that the finite volume quantum Gibbs states satisfy them. In Section 4, we discuss how the sates introduced in previous sections relate to the known KMS states for short-range interactions. 	
	\section{Preliminary Results}
	Following \cite{Gruber}, we introduce the state space as a group. Let $G_q \subset \mathbb{S}^1$ be the subgroup of all $q$ roots of unity, i.e.,
	\[
	G_q = \{e^{2\pi i\frac{k}{q}}: k=0,1, \dots, q-1\}.
	\]
	For simplicity, we choose $q=2$. With this assumption, we recover the usual state space $\{-1,+1\}$ for Ising spin systems. The Pontryagin dual of $G_q$ can be readily shown to be isomorphic to itself. Indeed, just notice that a character $\chi$ on $G_q$ must also be a root of unit since $1=\chi(1) = \chi(e^{\frac{2\pi i}{q}})^q $. Thus there must exist $k$ such that $\chi(e^{\frac{2\pi i}{q}}) = e^{2\pi i \frac{k}{q}}$. Moreover, the group $G_q$ is isomorphic to $\mathbb{Z}_q$, the group of integers modulo $q$. Define the action $\alpha: \mathbb{Z}_q \times G_q \rightarrow G_q$ given by
	\[
	\alpha(m, e^{2\pi i \frac{k}{q}}) = e^{2\pi i \frac{k+m}{q}}.
	\]
	With this action, we can define the transformation groupoid $G_q \rtimes_\alpha \mathbb{Z}_q$. This is a compact Hausdorff \'{e}tale groupoid with respect to the product topology. The following proposition follows by standard methods:
	\begin{proposition}
		$C^*(G_q \rtimes_\alpha \mathbb{Z}_q) \simeq M_q(\mathbb{C})$.
	\end{proposition}

	For each subset $\Lambda \subset \Z^d$, consider
	\[
	\Omega_\Lambda = \prod_{x\in \Lambda}\Z_n \quad \text{and} \quad \widehat{G_\Lambda} = \bigoplus_{x \in \Lambda} \mathbb{Z}_n.
	\]
	When $n=2$ we will consider $\Omega_{\Lambda} = \{-1,+1\}^\Lambda$ and use the multiplicative notation for $\mathbb{Z}_2$ since this is related to the more familiar case of two-valued spins in classical statistical mechanics. When $\Lambda = \Z^d$, we will write $\Omega_{\Z^d} \coloneqq \Omega$ and $G_{\Z^d}\coloneqq G$. Let the action $\alpha_\Lambda : \Omega_\Lambda \times G_\Lambda \rightarrow G_\Lambda$ be given by
	\[
	\alpha(\sigma,g) = (g_x\sigma_x)_{x\in \Lambda}, 
	\] 
	where we used implicitly that $g = \oplus_{x\in \Lambda}g_x$ and that $g_x\sigma_x$ is the usual action of $\Z_n$ in $\{0,1,\dots,n\}$. For the case $n=2$, this action is known as spin-flip. Our groupoid then will be the transformation groupoid 
	\[
	\mathcal{G}_\Lambda = \Omega_\Lambda \rtimes_{\alpha_\Lambda} G_\Lambda.
	\]
	Again, we will make the identification $\mathcal{G}\coloneqq \mathcal{G}_{\Z^d}$. We will denote the elements of the groupoid by different greek letters, to distinguish from the usual notation for configurations in classical statistical mechanics, e.g., $\sigma_\Lambda, \omega_\Lambda \in \Omega_\Lambda$ and $\bm{\sigma}_\Lambda, \bm{\omega}_\Lambda \in \mathcal{G}_\Lambda$. Some references for groupoid $\text{C}^*$-algebras  are \cite{putnam, Renault1980, Renault2,  SimsSzaboWilliams2020}. Let $\mathfrak{A}_n$ be the inductive limit C$^*$-algebra, constructed in section 6.2.1 of \cite{Bra2}. This algebra is also known as the UHF-algebra of type $n^\infty$. We denote by $\mathcal{P}_f(\Z^d)$ the set of finite subsets of $\Z^d$.
	\begin{theorem}
		The C$^*$ algebra $C^*(\mathcal{G})$ is isomorphic to the algebra $\mathfrak{A}_n$
	\end{theorem}

	\begin{definition}
		A function $\phi:\mathcal{P}_f(\Z^d)\rightarrow C^*(\mathcal{G})$ is called an interaction if $\phi(\Lambda)\coloneqq \phi_\Lambda = \phi_\Lambda^*$ and $\phi_\Lambda \in C^*(\mathcal{G}_\Lambda)$. An interaction is said to have short-range if there is $R>0$ such that if $\diam{X}>R$, then $\phi_X = 0$. Otherwise, it will be said that the interaction has long-range.
	\end{definition}

\begin{definition}
An interaction $\phi: \mathcal{P}_f(\mathbb{Z}^d)\rightarrow \mathfrak{U}$ is called classical when it is an interaction and $\phi_X \in C(\Omega_X)$ for all $X$. 
\end{definition}
	
	Before introducing the definition of boundary condition for the quantum statistical mechanics case, we first will derive a random representation for the Gibbs density operator. These representations appeared previously; see \cite{Aiz1, Aiz2, Gin, Fich1, Fich2, Iof, Rob} for a non-exhaustive list of papers where point processes were used to study quantum spin systems. In this case, the group $\bigoplus_{x\in \mathbb{Z}^d}\mathbb{Z}_2$ as the set $\mathcal{P}_f(\mathbb{Z}^d)$ with the product given by the symmetric difference. Let $\sigma_x^{(i)}$, $i=1,2,3$ be a copy of the Pauli matrices in the local algebra $\mathfrak{A}_x$. We will assume from now on that the interactions are homogeneous polynomials on the generators, i.e., 
	
	\[
	\phi_X = \sum_{ A\cup B = X } c_{A,B}\sigma_A^{(3)}\sigma_B^{(1)}.
	\]
	
	We will always use the representation where the Pauli operator $\sigma_A^{(3)}$ appears at the l.h.s of the operator $\sigma_B^{(1)}$. Since the interaction term $\phi_X$ must be self-adjoint we have that the constants must satisfy $\overline{c_{A,B}} = (-1)^{|A\cap B|}c_{A,B}$.
	We can write the coefficients as $c_{A,B} = r_{A,B}e^{i\pi \theta_{A,B}}$, where the number $\theta_{A,B}\in \{0,1, \pm 1/2\}$. The Hamiltonian operator is defined as
	\[
	H_\Lambda(\phi) \coloneqq \sum_{X \subset \Lambda }\phi_X.
	\]
	By definition of interaction, the Hamiltonian is a well-defined self-adjoint element of the algebra $C^*(\mathcal{G}_\Lambda)$. Following Bratteli and Robinson \cite{Bra2}, we define a surface term corresponding to the interaction between the region $\Lambda$ and the exterior region $\Lambda^c$
	\[
	W_\Lambda(\phi) \coloneqq \sum_{\substack{X \cap \Lambda \neq \emptyset \\ X \cap \Lambda^c \neq \emptyset}} \phi_X.
	\]
The surface energy term is not well defined for all possible interactions since, different from the Hamiltonian function $H_\Lambda$ described previously, it is potentially a sum of infinitely many terms, thus bringing convergence issues to the definition. Notice that this can be circumvented, for instance, if one assumes that the interaction has a short-range. For the general case, suitable decaying conditions can be made as a hypothesis to ensure convergence of the sum. We will focus only on the short-range case.  

	\begin{proposition}\label{poirep}
		The Gibbs density operator
		$e^{-\beta (H_\Lambda(\phi)+W_\Lambda(\phi))}$ has a  Poisson point process representation. 
	\end{proposition}
	\begin{proof}
		Let $\Z^d \sqcup \Z^d$ be the disjoint union of two copies of the lattice $\Z^d$. For each finite subset $X \subset \Z^d\sqcup \Z^d$ we can associate a pair of finite sets $A, B \subset \Z^d$ such that $X = A \sqcup B$. For $\Lambda \subset \Z^d$, define the set $F_\Lambda = \{X \in \mathcal{F}(\Z^d \sqcup \Z^d): (A\cup B)\cap \Lambda \neq \emptyset, B\neq \emptyset\}$. The Hamiltonian can be written as 
		\be
		H_\Lambda(\phi) + W_\Lambda(\phi)= H_\Lambda^{0}(\phi)+\sum_{X\in F_\Lambda} r_{X}S_X.
		\ee
		The operators $S_X$ above are defined as 
		\[
		S_X = e^{i\pi \theta_X}\sigma_A^{(3)}\sigma_B^{(1)} =  e^{i\pi \theta_X}\prod_{x \in A}\sigma_x^{(3)}\prod_{y \in B}\sigma_y^{(1)}, 
		\]
		where $\theta_X$ is the number that makes $S_X$ self-adjoint. The Hamiltonian $H_\Lambda^{(0)}(\phi)+W_\Lambda^{(0)}(\phi)$ is the classical part of the initial quantum Hamiltonian $H_\Lambda(\phi)+W_\Lambda(\phi)$, defined as
		\[
		H_\Lambda^{(0)}(\phi)+W_\Lambda^{(0)}(\phi) = \sum_{A\cap \Lambda \neq \emptyset} J_A\sigma_A^{(3)},
		\]
		where the constants $J_A = r_A e^{i\pi \theta_A}$ are real numbers, since the self-adjointeness of the interaction $\phi$ implies that $\theta_A=0$ or $1$. Remember that we are assuming that the interaction has a short-range, thus there are only finitely many constants $r_X$ that are different from $0$ for $X\in F_\Lambda$. The Lie-Trotter formula yields
		\be\label{eq1:poirep}
		e^{-\beta (H_\Lambda(\phi)+W_\Lambda(\phi))}= e^{\beta }\lim_{n\rightarrow \infty} \left[e^{-\frac{\beta}{n+1} (H_\Lambda^{(0)}(\phi)+W_\Lambda^{(0)}(\phi))}\left(\left(1-\frac{\beta}{n}\right)\mathbbm{1}+\frac{\beta}{n}V_{\phi,\Lambda}\right)\right]^n e^{-\frac{\beta}{n+1} (H_\Lambda^{(0)}(\phi)+W_\Lambda^{(0)}(\phi))},
		\ee
		where $V_{\phi,\Lambda} = -\sum_{X \in F_\Lambda}r_X S_X$. The sequence above can be expanded as
		
		\begin{align}\label{eq2:poirep}
		\left[e^{-\frac{\beta}{n+1} (H_\Lambda^{(0)}(\phi)+W_\Lambda^{(0)}(\phi))}\left(\left(1-\frac{\beta}{n}\right)\mathbbm{1}+\frac{\beta}{n}V_{\phi,\Lambda}\right)\right]^n& = \nonumber\\
		\sum_{j\in \{0,1\}^n} \prod_{m=1}^n e^{-\frac{\beta}{n} (H_\Lambda^{(0)}(\phi)+W_\Lambda^{(0)}(\phi))}&V_{\phi,\Lambda}^{j(m)}\left(1-\frac{\beta}{n+1}\right)^{n-\sum_{m=1}^n j(m)}\left(\frac{\beta}{n}\right)^{\sum_{m=1}^n j(m)}.
		\end{align}
		We can break the sum depending on each $j$ in the r.h.s of Equation \eqref{eq2:poirep} yielding us
		\begin{align*}
			\left[e^{-\frac{\beta}{n+1} (H_\Lambda^{(0)}(\phi)+W_\Lambda^{(0)}(\phi))}\left(\left(1-\frac{\beta}{n}\right)\mathbbm{1}+\frac{\beta}{n}V_{\phi,\Lambda}\right)\right]^n &e^{-\frac{\beta}{n+1} (H_\Lambda^{(0)}(\phi)+W_\Lambda^{(0)}(\phi))} = \\
			\sum_{\ell =0}^n\sum_{\substack{j\in \{0,1\}^n \\ |j(m)= 1|=\ell}} \left(\prod_{m=1}^n e^{-\frac{\beta}{n+1} (H_\Lambda^{(0)}(\phi)+W_\Lambda^{(0)}(\phi))}V_{\phi,\Lambda}^{j(m)}\right)&e^{-\frac{\beta}{n+1} (H_\Lambda^{(0)}(\phi)+W_\Lambda^{(0)}(\phi))}\left(1-\frac{\beta}{n}\right)^{n-\ell}\left(\frac{\beta}{n}\right)^{\ell}.
		\end{align*}
		Take some $j\in \{0,1\}^n$, where $|j(m)=1|=\ell$. Enumerate the points where $j$ is not zero into an increasing order $m_1 <\dots < m_\ell$. Hence,
		\begin{align*}
		\left(\prod_{m=1}^n e^{-\frac{\beta}{n+1} (H_\Lambda^{(0)}(\phi)+W_\Lambda^{(0)}(\phi))}V_{\phi,\Lambda}^{j(m)}\right)&e^{-\frac{\beta}{n+1} (H_\Lambda^{(0)}(\phi)+W_\Lambda^{(0)}(\phi))} = \\ &\left(\prod_{j=1}^\ell e^{-\frac{\beta(m_{j}-m_{j-1})}{n+1} (H_\Lambda^{(0)}(\phi)+W_\Lambda^{(0)}(\phi))}V_{\phi,\Lambda}\right)e^{-\frac{\beta(n+1-m_{\ell})}{n+1} (H_\Lambda^{(0)}(\phi)+W_\Lambda^{(0)}(\phi))}.
		\end{align*}

		Let $R$ be the range of the interaction $\phi$, and define $\Lambda_R = \{x \in \mathbb{Z}^d: \exists y \in \Lambda \text{  s.t.  } \|x-y\|\leq R\} $. We will calculate the value of the operator on the l.h.s of \eqref{eq1:poirep} in a point $(\omega_{\Lambda_R},\iota_X)$ and expand the r.h.s. using the  partition of identity
		\[
		\mathbbm{1} = \sum_{\omega_{\Lambda_R} \in \Omega_{\Lambda_R}}\delta_{\omega_{\Lambda_R}},
		\]
		where $\delta_{\omega_{\Lambda_R}}$ are delta functions on the unit $\omega_{\Lambda_R}\in \mathcal{G}^{(0)}_{\Lambda_R}$. These delta functions satisfy the property
		\be\label{eq3:poirep}
		\delta_{\omega_{\Lambda_R}} * f * \delta_{\eta_{\Lambda_R}} = f(\eta_{\Lambda_R},\iota_Y) \delta_{(\eta_{\Lambda_R},\iota_Y)},
		\ee
		where $Y$ is such that $\iota_Y\eta_{\Lambda_R}=\omega_{\Lambda_R}$, and $f\in C_c(\mathcal{G}_{\Lambda_R})$. Thus,
		\begin{align*}
			&\delta_{\iota_X\omega_{\Lambda_R}}*\left(\prod_{m=1}^n e^{-\frac{\beta}{n+1} (H_\Lambda^{(0)}(\phi)+W_\Lambda^{(0)}(\phi))}V_{\phi,\Lambda}^{j(m)}\right)e^{-\frac{\beta}{n+1} (H_\Lambda^{(0)}(\phi)+W_\Lambda^{(0)}(\phi))}*\delta_{\omega_{\Lambda_R}} =\\
			&\sum_{\substack{\omega_{{\Lambda_R},k} \in \Omega_{\Lambda_R} \\ 1\leq k \leq \ell}}\delta_{\iota_X\omega_{\Lambda_R}}*\left(\prod_{k=1}^\ell e^{-\frac{\beta(m_k -m_{k-1})}{n+1}(H_\Lambda^{(0)}(\phi)+W_\Lambda^{(0)}(\phi))(\omega_{{\Lambda_R},k})}\delta_{\omega_{{\Lambda_R},k}}*V_{\phi,\Lambda}\right)e^{-\frac{\beta(n+1-m_{\ell})}{n+1}(H_\Lambda^{(0)}(\phi)+W_\Lambda^{(0)}(\phi))}*\delta_{\omega_{\Lambda_R}} =\\
			&\sum_{\substack{\omega_{{\Lambda_R},k} \in \Omega_{\Lambda_R} \\ 1\leq k \leq \ell}}e^{-\beta E_\Lambda((\omega_{{\Lambda_R},1},m_1/(n+1)), \dots, (\omega_{{\Lambda_R},\ell},m_\ell/(n+1)))}\delta_{\iota_X\omega_{\Lambda_R}}*\left( \prod_{k=1}^\ell\delta_{\omega_{{\Lambda_R},k}}*V_{\phi,\Lambda}\right)*\delta_{\omega_{\Lambda_R}} ,
		\end{align*}
		where  the function $E_\Lambda$ is
		\[
		E_\Lambda((\omega_{{\Lambda_R},1},m_1/(n+1)), \dots, (\omega_{{\Lambda_R},\ell},m_\ell/(n+1)) = \frac{1}{n+1}\sum_{k=1}^{\ell+1} (m_k-m_{k-1})(H_\Lambda^{(0)}(\phi)+W_\Lambda^{(0)}(\phi))(\omega_{{\Lambda_R},k}),
		\]
		with $\omega_{\Lambda_R,\ell+1}=\omega_{\Lambda_R}$. We have
		\begin{align*}
			\delta_{\iota_X\omega_{\Lambda_R}}*\left( \prod_{k=1}^\ell\delta_{\omega_{{\Lambda_R},k}}*V_{\phi,\Lambda}\right)*\delta_{\omega_{\Lambda_R}}=\left(\prod_{k=1}^\ell V_{\phi,\Lambda}(\omega_{\Lambda_R,k},\iota_{X_k})\right)\delta_{\iota_X\omega_{\Lambda_R}=\omega_{{\Lambda_R},1}}\delta_{\omega_{\Lambda_R}},
		\end{align*}
		by the property \eqref{eq3:poirep}. Define the following function
		\[
		U_\Lambda(\omega_{{\Lambda_R},1},\dots,\omega_{{\Lambda_R},n};(\omega_{\Lambda_R},\iota_X)) = \prod_{k=1}^\ell V_{\phi,\Lambda}(\omega_{\Lambda_R,k},\iota_{X_k}).
		\]
		The r.h.s of \eqref{eq2:poirep} becomes
		\begin{align}\label{eq4:poirep}
			\sum_{\ell =0}^n\sum_{\substack{j\in \{0,1\}^n \\ |j(m)= 1|=\ell}}\sum_{\substack{\omega_{\Lambda,m} \in \Omega_\Lambda \\ 1\leq m \leq \ell}}&e^{-\beta E_\Lambda((\omega_{{\Lambda_R},1},m_1/n), \dots, (\omega_{{\Lambda_R},\ell},m_\ell/n))} \times \nonumber\\
			&U_\Lambda(\omega_{{\Lambda_R},1},\dots,\omega_{{\Lambda_R},\ell};(\omega_{\Lambda_R},\iota_X))\delta_{\iota_X\omega_{\Lambda_R}=\omega_{{\Lambda_R},1}}\left(1-\frac{\beta}{n}\right)^{n-\ell}\left(\frac{\beta}{n}\right)^{\ell}.
		\end{align}
		Consider the following Bernoulli point process
		\[
		N_n(x,C) = \sum_{j=1}^n \xi_{n,j}(x)\delta_{(\zeta_j(x),\frac{j}{n+1})}(C), 
		\]
		where $C$ is a Borel subset of $	\tilde{\Omega}_{{\Lambda_R}} = \Omega_{{\Lambda_R}}\times [0,1]$ and $\{\xi_{i,j}\}_{i \in \mathbb{N}, 1\leq j \leq n}$ and $\{\zeta_j\}_{j\in\mathbb{N}}$ are  families of i.i.d random variables with distribution
		\[
		\mathbb{P}(\xi_{n,j}=1) = 1- \mathbb{P}(\xi_{n,j}=0)= \frac{\beta}{n}\quad \text{  and  } \quad \mathbb{P}(\zeta_j=\omega_{\Lambda_R})= 1.
		\] 
		Defining the time ordering functions $T_n:[0,1]^n \rightarrow [0,1]^n$
		\[
		T_n(t_1,\dots,t_n) = (t'_1,\dots, t'_n),
		\]
		where the r.h.s is a permutation of $(t_1,\dots,t_n)$ satisfying $t'_i \leq t'_{i+1}$, $i=1,\dots, n-1$. One can write the function $T_n$ more explicity using characteristics functions in the following way
		\[
		T_n(t_1,\dots,t_n) = \sum_{p\in\mathfrak{S}_n}\left(\prod_{i=1}^{n-1} \mathbbm{1}_{\{x>0\}}(t_{p(i)}-t_{p(i+1)})\right)(t_{p(1)},\dots,t_{p(n)}),
		\] 
		where $\mathfrak{S}_n$ is the permutation group of $n$ points. Notice that this function is zero whenever we have two coordinates $t_i$ and $t_j$ that are equal. Since the set of points where there are two equal coordinates have measure zero on $[0,1]^n$ with respect to the Lebesgue measure, we can redefine it to give a nonzero value. Using the functions $T_n$ we can construct a time ordering function $T$ on the coproduct (See Appendix A). By a similar procedure, by introducing the functions $f_n:\tilde{\Omega}_\Lambda^n \rightarrow \mathbb{R}$ given by
		\[
		f_n((\omega_{{\Lambda_R},i},t_i)_n;(\omega_{\Lambda_R},\iota_X)) = e^{-\beta E_\Lambda((\omega_{{\Lambda_R},i},t_i)_n;(\omega_{\Lambda_R},\iota_X))}U_\Lambda((\omega_{{\Lambda_R},i},t_i)_n;(\omega_{\Lambda_R},\iota_X)),
		\]
		where $(\omega_{\Lambda_R,i},t_i)_n = ((\omega_{{\Lambda_R},1},t_1),\dots,(\omega_{{\Lambda_R},n},t_n))$ with
		\begin{align*}
		&E_\Lambda((\omega_{\Lambda_R,i},t_i)_n;(\omega_{\Lambda_R},\iota_X)) = \sum_{i=0}^n (t'_{i+1}-t'_i)(H_\Lambda^{(0)}(\phi)+W_\Lambda^{(0)}(\phi))(\omega_{{\Lambda_R},i}) \quad \text{and}\quad \\
		 &U_\Lambda((\omega_{{\Lambda_R},i},t_i)_n;(\omega_{\Lambda_R},\iota_X)) = \prod_{k=1}^n V_{\phi,\Lambda}(\omega_{\Lambda_R,k},\iota_{X_k}),
		\end{align*}
		
		where $t_0 = 0$, $t_{n+1}=1$, $\omega_{\Lambda_R,0} = \iota_X \omega_{\Lambda_R}$, and $\omega_{\Lambda_R,n} = \omega_{\Lambda_R}$ we can create a function $f$ defined on the coproduct of $\tilde{\Omega}_{\Lambda_R}^n$, for $n\geq 0$. The expression \eqref{eq4:poirep} is then, an integral of the function $f$ with respect to the Binomial point process. Since the Binomial point process converges in distribution to a Poisson point process we have that 
		\[
		e^{-\beta (H_\Lambda(\phi)+W_\Lambda(\phi))}(\omega_{\Lambda_R},\iota_X) = \lim_{n\rightarrow \infty}\int_{\Omega} f \circ N_n(\omega) d\mathbb{P}(\omega) = \int_\Omega f \circ N (\omega) d\mathbb{P}(\omega) 
		\]
		
	\end{proof}
	
	We know that the function $U_\Lambda$ can be expanded by using the fact that $V_{\phi,\Lambda}= \sum_{X\in F_\Lambda} r_X S_X$, in the following way
	\begin{align*}
	U_\Lambda((\omega_{\Lambda_R,i},t_i)_n;(\omega_{\Lambda_R},\iota_X)) &= \prod_{k=1}^n V_{\phi,\Lambda}(\omega_{\Lambda_R,k},\iota_{X_k}) \\
	& = \prod_{k=1}^n\left(\sum_{X \in F_\Lambda}r_X S_X(\omega_{\Lambda_R,k},\iota_{X_k})\right) \\
	& = \sum_{\substack{X_j \in F_\Lambda \\ 1\leq j \leq n}} \prod_{k=1}^n r_{X_j}S_{X_j}(\omega_{\Lambda_R,k},\iota_{X_k}).
	\end{align*}
	The Poisson point process representation in this case is a rigorous path integral for quantum spin systems. We will proceed now to write in a way that the analogy with the paths is more transparent. Using the integration formula for Poisson point processes in Proposition \ref{integformpoirep:app1}, we know
	
	\begin{align*}
	\int_{\mathbb{N}(X)} f dN &= \sum_{n\geq 0} \frac{\beta^n}{n!}\int_{[0,1]^n}\sum_{\omega_{\Lambda_R} \in \Omega_{\Lambda_R}^{n-1}}e^{-\beta E_\Lambda((\omega_{\Lambda_R,i},t_i)_n;(\omega_{\Lambda_R},\iota_X))}U_\Lambda((\omega_{\Lambda_R,i},t_i)_n;(\omega_{\Lambda_R},\iota_X))dt^n \\
	&= \sum_{n\geq 0} \frac{\beta^n}{n!}\int_{[0,1]^n}\sum_{\omega_{\Lambda_R} \in \Omega_{\Lambda_R}^{n-1}}\sum_{\substack{X_j \in F_\Lambda \\ 1\leq j \leq n}}e^{-\beta E_\Lambda((\omega_{\Lambda_R,i},t_i)_n;(\omega_{\Lambda_R},\iota_X))} \prod_{k=1}^n r_{X_j}S_{X_j}(\omega_{\Lambda_R,k},\iota_{X_k})dt^n.
	\end{align*}
	By Corollary \ref{Corol_ppp}, the formula above is just the integration of the map
	\[
	(\omega_{\Lambda_R,i},t_i,X_i)_n \mapsto  e^{-\beta E_\Lambda((\omega_{\Lambda_R,i},t_i)_n;(\omega_{\Lambda_R},\iota_X))} \prod_{k=1}^n S_{X_j}(\omega_{\Lambda_R,k},\iota_{X_k}),
	\]
	with respect to the Poisson Point process $N_{\Lambda_R}\coloneqq \sum_{X\in F_\Lambda} N_X$, where each $N_X$ is a Poisson point process with intensity measure $\beta r_X dt$. The point process $N_{\Lambda_R}$ can be decomposed as the sum of two independent Poisson point processes
	\[
	N_{\Lambda_R}  = \sum_{\substack{X \subset \Lambda \\ X \in F_\Lambda}} N_X + \sum_{\substack{X \cap \Lambda_R\setminus \Lambda \neq \emptyset \\ X \in F_\Lambda}} N_X = N_\Lambda + N_{\Lambda_R\setminus\Lambda}.
	\] 
	
	\begin{lemma}\label{lemma_ppp_decomp}
		Let $f:\mathbb{N}(\tilde{\Omega}_{\Lambda_R} \times F_\Lambda)\rightarrow \mathbb{R}$ be a bounded measurable function. It holds
		\be\label{pp_decomp}
		\int f(\nu) d N_{\Lambda_R}(\nu) = \int\int f(\nu+\nu') dN_{\Lambda}(\nu) dN_{\Lambda_R\setminus\Lambda}(\nu')
		\ee
	\end{lemma}
	\begin{proof}
		Corollary \ref{Corol_ppp} gives us the following formula for the innermost integral in the r.h.s of Equation \eqref{pp_decomp}
		\begin{align*}
			\int f(\nu+\nu') dN_{\Lambda}(\nu) = \sum_{k\geq 0}\frac{\beta^k}{k!}\int_{[0,1]^k}\sum_{\substack{\sigma_{\Lambda_R,i}\in \Omega_{\Lambda_R}^k \\ X_i\subset \Lambda, X_i\in F_\Lambda}} f\left(\sum_{i=1}^k\delta_{(\sigma_{\Lambda_R,i},t_i,X_i)}+\nu'\right)\prod_{i=1}^k r_{X_i}dt^k.
		\end{align*}
	Using the linearity of the integral and again using the formula given in Corollary \ref{Corol_ppp}, we get
		\begin{align*}
		&\int\left(\int_{[0,1]^k}\sum_{\substack{\sigma_{\Lambda_R,i}\in \Omega_{\Lambda_R}^n \\ X_i\subset \Lambda, X_i\in F_\Lambda}} f\left(\sum_{i=1}^k\delta_{(\sigma_{\Lambda_R,i},t_i,X_i)}+\nu'\right)\prod_{i=1}^k r_{X_i}dt^k\right) dN_{\Lambda_R\setminus\Lambda}(\nu') = \\
		&\sum_{m\geq 0}\frac{\beta^m}{m!}\int_{[0,1]^m}\int_{[0,1]^k}\sum_{\substack{\sigma_{\Lambda_R,i}\in \Omega_{\Lambda_R}^k \\ X_i\subset \Lambda, X_i\in F_\Lambda}}\sum_{\substack{\omega_{\Lambda_R,i}\in \Omega_{\Lambda_R}^m \\ Y_i\cap \Lambda_R\setminus\Lambda\neq \emptyset, Y_i\in F_\Lambda}} f\left(\sum_{i=1}^k\delta_{(\sigma_{\Lambda_R,i},t_i,X_i)}+\sum_{i=1}^m\delta_{(\omega_{\Lambda_R,i},t_i,Y_i)}\right)\prod_{i=1}^k r_{X_i}\prod_{i=1}^m r_{Y_i}dt^kdt^m. 
		\end{align*}
	Thus, by rearranging the terms, we can write the total integral as the sum over $n\geq 0$ of all $k+m = n$ terms 
		\begin{align*}
			&\int\int f(\nu+\nu') dN_{\Lambda}(\nu) dN_{\Lambda_R\setminus\Lambda}(\nu') = \\
			&\sum_{n\geq 0}\frac{\beta^n}{n!}\sum_{k+m=n}\frac{n!}{k!m!}\int_{[0,1]^k}\int_{[0,1]^m}\sum_{\substack{\sigma_{\Lambda_R,i}\in \Omega_{\Lambda_R}^k, \omega_{\Lambda_R,i}\in \Omega_{\Lambda_R}^m \\ X_i\subset \Lambda, X_i\in F_\Lambda \\ Y_i\cap \Lambda_R\setminus\Lambda\neq \emptyset, Y_i\in F_\Lambda}} f\left(\sum_{i=1}^k\delta_{(\sigma_{\Lambda_R},t,X_i)}+\sum_{i=1}^m\delta_{(\omega_{\Lambda_R},t,Y_i)}\right)\prod_{i=1}^k r_{X_i}\prod_{i=1}^m r_{Y_i}dt^kdt^m.
		\end{align*}
			We can break the hypercube $[0,1]^n$ using the subsets $A\subset \{1,\dots,n\}$ in the following way
			\[
			[0,1]^n = \bigcup_{A\subset \{1,\dots,n\}}\{(t_{i_1},\dots,t_{i_p}):i_j\in A\}\times \{(t_{i_1},\dots,t_{i_q}):i_j\in A^c\},
			\]
			thus
			\[
			\int_{[0,1]^n} =\sum_{k=0}^n \binom{n}{k}\int_{[0,1]^{n-k}}\int_{[0,1]^{k}}dt^{n-k}dt^k,
			\]
			yielding us
			\begin{align*}
			\int\int f(\nu+\nu') dN_{\Lambda}(\nu) dN_{\Lambda_R\setminus\Lambda}(\nu') = \sum_{n\geq 0}\frac{\beta^n}{n!}\int_{[0,1]^n}\sum_{\substack{\sigma_{\Lambda_R,i}\in \Omega_{\Lambda_R}^n\\ X_i\in F_\Lambda}} f\left(\sum_{i=1}^n\delta_{(\sigma_{\Lambda_R},t,X_i)}\right)\prod_{i=1}^n r_{X_i}dt^n.
			\end{align*}
		Corollary \ref{Corol_ppp} allows us to finish the proof. 
	\end{proof}
	
	The random representation for the Gibbs density operator has the defining feature of always coming with a preferred permutation for the times for the arrivals: it is always in increasing order, realized by the composition with the time ordering function $T$ introduced earlier. Hence a good way of interpreting this random representation is through the notion of what is known as a \emph{path integral}. This interpretation could be made rigorous through the definition of a measure on the space of cadlag functions, as in \cite{Bil}. We preferred the representation presented here, as an integral over the space of point measures, since we found it more simple. Nonetheless, the integral over paths idea motivates us to introduce the following notations for the functions being integrated, but first, a few observations must be made. Notice that the functions $S_X$ satisfy
		\[
	S_X((\omega_{\Lambda_R},\iota_X)) = \begin{cases}
		e^{i\pi\theta_X}\prod_{x\in A}(\iota_B\omega_{\Lambda_R})_x & X=B \\
		0 & o.w.
	\end{cases}
	\]
	Thus not every point measure $\alpha_{\Lambda_R}\in \mathbb{N}(\tilde{\Omega}_{\Lambda_R}\times \mathcal{P}(\Lambda_R \sqcup \Lambda_R))$ will contribute to the integral representation, only those that, after the time ordering operation, are coherent with respect to the sets appearing in the jumps. We can make this rigorous by introducing the following subset of $\mathbb{N}(\tilde{\Omega}_\Lambda\times \mathcal{P}(\Lambda_R\sqcup \Lambda_R))$
	\be\label{path_def}
	\mathcal{P}_{\Lambda_R}^{\omega_{\Lambda_R},X} \coloneqq \left\{\alpha_{\Lambda_R} = \sum_{i=1}^n\delta_{(\omega_{\Lambda_R,i},t_i,X_i)}, n\geq 0: \omega_{\Lambda_R,1} = \omega_\Lambda, \omega_{\Lambda_R,n}= \iota_X\omega_\Lambda, \omega_{\Lambda_R, i-1} = \iota_{B_i}\omega_{\Lambda_R,i} , 1\leq i\leq n\right\},
	\ee
	Given another point measure $\xi_{\Lambda_R}$, we can also introduce the following set 
	\be\label{path_bc_def}
	\mathcal{P}_{\Lambda_R}^{\omega_{\Lambda_R},X,\xi_{\Lambda_R}} \coloneqq \xi_{\Lambda_R}+\mathcal{P}_{\Lambda_R}^{\omega_{\Lambda_R},X}.
	\ee
	Although some elements of the set above may not generate coherent paths, it will be clear that in our applications it will always be the case. Each point measure of $\mathcal{P}_{\Lambda_R}^{\omega_{\Lambda_R},X}$ can be viewed as a path by rearranging the jumps following the increasing order of the time, i.e.,
	\[
	\alpha_{\Lambda_R}(t) = \omega_{\Lambda_R,i}, \quad \text{ if} \;\; t_{i-1}\leq t <t_i,
	\]
	where we are supposing, in the definition above, that the times $t_i$ are already ordered, see the Figure 

	
	We will use, for the path, the same notation as in the point measure. Finally, the following notation will be used from now on
		\begin{align*}
	&E_{\Lambda_R}((\omega_{\Lambda_R,i},t_i)_n;(\omega_{\Lambda_R},\iota_X)) = \sum_{i=0}^n (t'_{i+1}-t'_i)H_{\Lambda_R}^{(0)}(\phi)(\omega_{{\Lambda_R},i}) =\int_{\alpha_{\Lambda_R}} H_{\Lambda_R}^{(0)}(\phi) \quad \text{and}\quad \\
	&S_{\Lambda_R}(\alpha_R) = \prod_{i=1}^n e^{i\pi \theta_{X_i}}\prod_{x\in A_i}(\iota_{B_i}\omega_{\Lambda,i})_{x}.
	\end{align*}
	For every bounded measurable function $f:\mathbb{N}(\tilde{\Omega}_{\Lambda_R}\times F_\Lambda)\rightarrow \mathbb{C}$, let us introduce the measure 
	\be\label{equation_measure_definition}
	\int_{\mathcal{P}_{\Lambda_R}^{\omega_{\Lambda_R},X}} f(\alpha_{\Lambda_R}) d\nu_{\phi,\Lambda_R}(\alpha_{\Lambda_R})  = \sum_{n\geq 0}\frac{\beta^n}{n!}\int_{[0,1]^n}\sum_{\substack{\omega_{\Lambda_R,i}\in \Omega_{\Lambda_R}, \\ X_i \in F_{\Lambda_R}}}f(\omega_{\Lambda_R,i},t_i, X_i) 
	\prod_{i=1}^nr_{X_i}dt^n, 
	\ee
	Notice that this measure can be related to the integration with respect to the Poisson point process $N_{\Lambda_R}$, introduced previously in this section. We get the expression,
	\[
	e^{-\beta H_{\Lambda_R}(\phi)}(\omega_{\Lambda_R},X)= e^{\beta}\int_{\mathcal{P}_{\Lambda_R}^{\bm{\omega}_{\Lambda_R},X}} e^{-\beta \int_{\alpha_{\Lambda_R}} H_{\Lambda_R}^{(0)}(\phi)}S(\alpha_{\Lambda_R})d\nu_{\phi,\Lambda_R}(\alpha_{\Lambda_R}).
	\]
	
	\section{The DLR states and the Gibbs states}
	
	Another definition that is more suitable for classical systems was introduced by Dobrushin \cite{Dob} and Lanford and Ruelle \cite{Lan} and nowadays is called the DLR equation. 
\begin{definition}
Let $\Omega = \{E\}^{\mathbb{Z}^d}$, where $|E| < \infty$, with the product topology, and consider $\mathcal{F}$ the Borel $\sigma$-algebra. Consider $\mathcal{F}_{\Lambda^c}$, for $\Lambda$ a finite subset of $\mathbb{Z}^d$, the smallest $\sigma-algebra$ such that $\{\pi_i: i \in \Lambda^c\}$, the projections, are measurable. Let $\mu$ be a probability measure on $(\Omega, \mathcal{F})$. We say that $\mu$ satisfies the DLR equations if, and only if, for all $f \in C(\Omega)$, we have:
\[
\mathbb{E}_\mu(f|\mathcal{F}_{\Lambda^c})(\omega) = \mu_{\Lambda, \beta}^\omega(f)\;\; \mu\;\; a.e.
\]
Where $\mathbb{E}_\mu(\cdot|\mathcal{F}_{\Lambda^c})$ is the conditional expectation with respect to the $\sigma$-algebra $\mathcal{F}_{\Lambda^c}$, and $ \mu_{\Lambda, \beta}^\omega$ is the local Gibbs measure.
\end{definition}
For more on DLR equations, with more details and in more general contexts, see \cite{Geo,Vel}. 
	The quantum Gibbs state for the interaction given by the Hamiltonian $H_{\Lambda_R}(\phi)$ is given by the following expression
	\[
	\mu_{\beta,\phi,\Lambda_R}(f) = \frac{1}{Z_{\beta,\phi,\Lambda_R}}\sum_{\sigma_{\Lambda_R}\in \Omega_{\Lambda_R}}f*e^{-\beta H_{\Lambda_R}(\phi)}(\sigma_{\Lambda_R}),
	\]
	where $Z_{\beta,\phi,\Lambda} = \sum_{\sigma_{\Lambda_R}}e^{-\beta H_{\Lambda_R}(\phi)}(\sigma_{\Lambda_R})$ is the partition function. We proceed to give a special decomposition of this state in terms of the Poisson point process $N_{\Lambda_R}$. By the definition of the product in a groupoid, we get
	\begin{align*}
		f* e^{-\beta H_{\Lambda_R}(\phi)}(\sigma_{\Lambda_R}) &= \sum_{(\omega_{\Lambda_R},X)\in \mathcal{G}_{\Lambda_R}^{\sigma_{\Lambda_R}}} f(\omega_{\Lambda_R},X)e^{-\beta H_{\Lambda_R}(\phi)}(\sigma_{\Lambda_R}, X) \\
		&= e^{\beta}\sum_{(\omega_{\Lambda_R}, X)\in \mathcal{G}_{\Lambda_R}^{\sigma_{\Lambda_R}}} \int_{\mathcal{P}_{\Lambda_R}^{\sigma_{\Lambda_R},X}} f(\omega_{\Lambda_R},X)e^{-\beta\int_{\alpha_{\Lambda_R}} H_{\Lambda_R}^{(0)}(\phi)}S_{\Lambda_R}(\alpha_{\Lambda_R})d\nu_{\phi,\Lambda_R}(\alpha_{\Lambda_R})
	\end{align*}
	Thus
	
	\[
	\mu_{\beta,\phi,\Lambda_R}(f) = \frac{1}{e^{-\beta}Z_{\beta,\phi,\Lambda_R}}\sum_{\sigma_{\Lambda_R}\in \Omega_{\Lambda_R}}\sum_{(\omega_{\Lambda_R},X)\in \mathcal{G}_{\Lambda_R}^{\sigma_{\Lambda_R}}} \int_{\mathcal{P}_{\Lambda_R}^{\sigma_{\Lambda_R},X}} f(\omega_{\Lambda_R},X)e^{-\beta\int_{\alpha_{\Lambda_R}} H_{\Lambda_R}^{(0)}(\phi)}S_{\Lambda_R}(\alpha_{\Lambda_R})d\nu_{\phi,\Lambda_R}(\alpha_{\Lambda_R})
	\]
	where we can write the partition function in the following representation
	\[
	e^{-\beta}Z_{\beta,\phi,\Lambda_R} = \sum_{\sigma_{\Lambda_R}\in \Omega_{\Lambda_R}}\int_{\mathcal{P}_{\Lambda_R}^{\sigma_{\Lambda_R}}} e^{-\beta\int_{\alpha_{\Lambda_R}} H_{\Lambda_R}^{(0)}(\phi)}S_{\Lambda_R}(\alpha_{\Lambda_R})d\nu_{\phi,\Lambda_R}(\alpha_{\Lambda_R})
	\]	 
	Although the above expression seems cumbersome at first sight, it satisfies a decomposition property similar to the DLR equations. To show this, first, notice that the classical Hamiltonian $H_{\Lambda_R}^{(0)}(\phi)$ can be written as
	\begin{align*}
		H_{\Lambda_R}^{(0)}(\phi)= \sum_{A \cap \Lambda \neq \emptyset} J_A \sigma_{A}^{(3)} +\sum_{A \subset \Lambda_R\setminus\Lambda}J_A \sigma_A^{(3)} = H_\Lambda^{(0)}(\phi)+W_\Lambda^{(0)}(\phi) + H_{\Lambda_R\setminus\Lambda}^{(0)}(\phi).
	\end{align*}
	Yielding us
	\begin{align*}
		\int_{\alpha_{\Lambda_R}} H_{\Lambda_R}^{(0)}(\phi) = \int_{\alpha_{\Lambda_R}}(H_\Lambda^{(0)}(\phi)+W_\Lambda^{(0)}(\phi))+\int_{\alpha_{\Lambda_R}} H_{\Lambda_R\setminus\Lambda}^{(0)}(\phi).
	\end{align*}
	
	Thus, we can break a path $\alpha_\Lambda$ into two paths, one with jumps only on sets $B \cap \Lambda_R\setminus\Lambda \neq \emptyset$, that we will call $\alpha_{\Lambda_R\setminus\Lambda}$, and the jumps occurring only inside $\Lambda$, that we will call $\alpha_{\Lambda}$. Notice that since there are no changes in the points $x \in \Lambda_R\setminus\Lambda$ configurations when the jump occurs at times in $\alpha_{\Lambda}$, this implies that
	\[
	\int_{\alpha_{\Lambda_R}} H_{\Lambda_R\setminus\Lambda}^{(0)}(\phi) = \int_{\alpha_{\Lambda_R\setminus\Lambda}} H_{\Lambda_R\setminus\Lambda}^{(0)}(\phi).
	\] 
	Since we cannot remove the dependence of the Hamiltonian in the jumps occurring in the path $\alpha_{\Lambda_R\setminus\Lambda}$, we will write the other term as
	\[
	\int_{\alpha_{\Lambda_R}}(H_\Lambda^{(0)}(\phi)+W_\Lambda^{(0)}(\phi)) = \int_{\alpha_\Lambda}H_{\Lambda}^{\alpha_{\Lambda_R\setminus\Lambda}}(\phi)
	\]
	A similar decomposition is available for the function $S_{\Lambda_R}$,
	\begin{align*}
		S_{\Lambda_R}(\alpha_{\Lambda_R}) = \left(\prod_{i:B_i\cap\Lambda\neq \emptyset}e^{i\pi \theta_{X_i}}\prod_{x\in A_i}(\iota_{B_i}\omega_{\Lambda,i})_{x}\right)\left(\prod_{i:B_i\subset \Lambda_R\setminus \Lambda}e^{i\pi \theta_{X_i}}\prod_{x\in A_i}(\iota_{B_i}\omega_{\Lambda,i})_{x}\right) = S_\Lambda^{\alpha_{\Lambda_R\setminus\Lambda}}(\alpha_{\Lambda})S_{\Lambda_R\setminus\Lambda}(\alpha_{\Lambda_R\setminus\Lambda}).
	\end{align*}
	Due to Lemma \ref{lemma_ppp_decomp}, the following decomposition is possible
	\begin{align*}
		&e^{-\beta H_{\Lambda_R}(\phi)}(\omega_{\Lambda_R},X) = \\
		&\int_{\mathcal{P}_{\Lambda_R\setminus\Lambda}^{\omega_{\Lambda_R\setminus\Lambda},X\setminus\Lambda}}\left(\int_{\mathcal{P}_{
				\Lambda}^{\omega_\Lambda,X\cap \Lambda,\alpha_{\Lambda_R\setminus\Lambda}}}e^{-\beta \int_{\alpha_\Lambda}H_\Lambda^{\alpha_{\Lambda_R\setminus \Lambda}}(\phi)}S_{\Lambda}^{\alpha_{\Lambda_R\setminus\Lambda}}(\alpha_\Lambda)d\nu_{\phi,\Lambda}\right)e^{-\beta\int_{\alpha_{\Lambda_R\setminus\Lambda}}H_{\Lambda_R\setminus\Lambda}^{(0)}(\phi)}S_{\Lambda_R\setminus\Lambda}(\alpha_{\Lambda_R\setminus\Lambda})d\nu_{\phi,\Lambda_R\setminus\Lambda}.
	\end{align*}
	Thus by calling the innermost integral in the equation above
	\[
	D_{\beta,\Lambda}^{\alpha_{\Lambda_R\setminus\Lambda}}(\omega_\Lambda,X\cap\Lambda) = \int_{\mathcal{P}_{\Lambda}^{\omega_\Lambda,X\cap\Lambda,\alpha_{\Lambda_R\setminus\Lambda}}}e^{-\beta \int_{\alpha_\Lambda}H_\Lambda^{\alpha_{\Lambda_R\setminus\Lambda}}}S_\Lambda^{\alpha_{\Lambda_R\setminus\Lambda}}(\alpha_\Lambda)d\nu_{\phi,\Lambda}(\alpha_\Lambda),
	\]
	we get the following representation for the point process
	\[
	e^{-\beta H_{\Lambda_R}(\phi)}(\omega_{\Lambda_R},X) = \int_{\mathcal{P}_{\Lambda_R\setminus\Lambda}^{\omega_{\Lambda_R\setminus\Lambda},X\setminus\Lambda}}D_{\beta,\Lambda}^{\alpha_{\Lambda_R\setminus\Lambda}}(\omega_\Lambda,X\cap\Lambda)e^{-\beta\int_{\alpha_{\Lambda_R\setminus\Lambda}}H_{\Lambda_R\setminus\Lambda}^{(0)}(\phi)}S_{\Lambda_R\setminus\Lambda}(\alpha_{\Lambda_R\setminus\Lambda})d\nu_{\phi,\Lambda_R\setminus\Lambda}.
	\]
	The following expression for the product with a function $f \in C_c(\mathcal{G}_{\Lambda_R})$ is possible
	\begin{align*}
		&\sum_{\sigma_{\Lambda_R}\in \Omega_{\Lambda_R}}f* e^{-\beta H_{\Lambda_R}(\phi)}(\sigma_{\Lambda_R}) = \\
		&\sum_{\sigma_{\Lambda_R}\in \Omega_{\Lambda_R} }\sum_{(\omega_{\Lambda_R},X)\in \mathcal{G}_{\Lambda_R}^{\sigma_{\Lambda_R}}} \int_{\mathcal{P}_{\Lambda_R\setminus\Lambda}^{\sigma_{\Lambda_R\setminus\Lambda},X\setminus\Lambda}}f(\omega_{\Lambda_R},X)D_{\beta,\Lambda}^{\alpha_{\Lambda_R\setminus\Lambda}}(\sigma_{\Lambda},X\cap\Lambda)e^{-\beta\int_{\alpha_{\Lambda_R\setminus\Lambda}}H_{\Lambda_R\setminus\Lambda}^{(0)}(\phi)}S_{\Lambda_R\setminus\Lambda}(\alpha_{\Lambda_R\setminus\Lambda})d\nu_{\phi,\Lambda_R\setminus\Lambda}(\alpha_{\Lambda_R\setminus\Lambda}) 
	\end{align*}
	
	We can break the sum above into two sums, the first one summing over the arrows of $\mathcal{G}_{\Lambda_R\setminus\Lambda}^{\sigma_{\Lambda_R\setminus\Lambda}}$ and the second over $\mathcal{G}_\Lambda^{\sigma_\Lambda}$. A similar procedure can be made with the sum depending on the configuration space. Thus, by defining the following expression
	\be\label{f_gibbs_path}
	\mu^{\alpha_{\Lambda_R\setminus\Lambda}}_{\beta,\phi,\Lambda}(f)(\omega_{\Lambda_R\setminus\Lambda},X\setminus\Lambda) = \frac{1}{Z_{\beta,\phi,\Lambda}^{\alpha_{\Lambda_R\setminus \Lambda}}}\sum_{\sigma_\Lambda \in \Omega_\Lambda}f* D_{\beta,\Lambda}^{\alpha_{\Lambda_R\setminus\Lambda}}(\sigma_\Lambda),
	\ee
	where the normalization is given by the partition function
	\[
	Z_{\beta,\phi,\Lambda}^{\alpha_{\Lambda_R\setminus \Lambda}} = \sum_{\sigma_\Lambda \in \Omega_{\Lambda}}D_{\beta,\Lambda}^{\alpha_{\Lambda_R\setminus \Lambda}}(\sigma_\Lambda)
	\]
	Plugging this again in the equation, we get
	
	\begin{align*}
		&\sum_{\sigma_{\Lambda_R}\in \Omega_{\Lambda_R}}f* e^{-\beta(H_\Lambda(\phi)+W_\Lambda(\phi))}(\sigma_{\Lambda_R})= \\
		&\sum_{\sigma_{\Lambda_R}\in \Omega_{\Lambda_R} }\sum_{(\omega_{\Lambda_R},X)\in \mathcal{G}_{\Lambda_R}^{\sigma_{\Lambda_R}}} \int_{\mathcal{P}_{\Lambda_R\setminus\Lambda}^{\sigma_{\Lambda_R\setminus\Lambda},X\setminus\Lambda}}Z_{\beta,\phi,\Lambda}^{\alpha_{\Lambda_R\setminus \Lambda}}\mu^{\alpha_{\Lambda_R\setminus\Lambda}}_{\beta,\phi,\Lambda}(f)(\omega_{\Lambda_R\setminus\Lambda},X\setminus\Lambda)e^{-\beta\int_{\alpha_{\Lambda_R\setminus\Lambda}}G(\phi)}T_\Lambda(\alpha_{\Lambda_R\setminus\Lambda})d\nu_{\phi,\Lambda_R\setminus\Lambda}(\alpha_{\Lambda_R\setminus\Lambda})= \\
		&\sum_{\sigma_{\Lambda_R}\in \Omega_{\Lambda_R} }\sum_{(\omega_{\Lambda_R},X)\in \mathcal{G}_{\Lambda_R}^{\sigma_{\Lambda_R}}} \int_{\mathcal{P}_{\Lambda_R}^{\sigma_{\Lambda_R},X}}\mu^{\alpha_{\Lambda_R\setminus\Lambda}}_{\beta,\phi,\Lambda}(f)(\omega_{\Lambda_R},X)e^{-\beta\int_{\alpha_{\Lambda_R}}H^{(0)}(\phi)}S_\Lambda(\alpha_{\Lambda_R})d\nu_{\phi,\Lambda_R\setminus\Lambda}(\alpha_{\Lambda_R\setminus\Lambda}).
	\end{align*}
	
	where 
	
	\[
	\mu^{\alpha_{\Lambda_R\setminus\Lambda}}_{\beta,\phi,\Lambda}(f)(\omega_{\Lambda_R},X) = \mathbbm{1}(\omega_{\Lambda},X\cap \Lambda)\mu^{\alpha_{\Lambda_R\setminus\Lambda}}_{\beta,\phi,\Lambda}(f)(\omega_{\Lambda_R\setminus\Lambda},X\setminus\Lambda).
	\]
	Yielding us a decomposition similar to the one encountered in the usual classical case for the DLR equations. So, for this reason we will use the expression for the finite Gibbs functional as the one defined by Equation \eqref{f_gibbs_path}. But there is a problem with this functional, related to the way the densities $D_{\beta,\Lambda}^{\alpha_{\Lambda_R\setminus\Lambda}}$ are defined. 
	From now on, we will refer to the path $\alpha_{\Lambda_R\setminus\Lambda}$ as $\alpha_{\Lambda^c}$, in order to lighten the notation. The concatenation of paths can only be defined on paths $\alpha_{\Lambda}$ $\alpha'_{\Lambda}$
	where $\alpha_{\Lambda}(1)=\alpha'_{\Lambda}(0)$, by the following composition rule
	\[
	\alpha_\Lambda \circ \alpha'_{\Lambda}(t) =  \begin{cases}
		\alpha_\Lambda(2t) & 0\leq t\leq 1/2 \\
		\alpha'_\Lambda(2t-1) & 1/2 \leq t \leq 1.
	\end{cases}
	\]
	We can also define an involutive operation that just reverses the path. This is, given a path $\alpha_\Lambda$, we define $\alpha_\Lambda^{-1}$ by $\alpha^{-1}_\Lambda(t)= \alpha_\Lambda(1-t)$. We are ready to prove the next lemma:
	\begin{lemma}\label{lemma_property_path}
		For every $\alpha_{\Lambda^c}$ and $\alpha'_{\Lambda^c}$ such that $\alpha_{\Lambda^c}(1)=\alpha'_{\Lambda^c}(0)$ and $\beta,\beta'$ it holds that 
		\[
		D_{\beta,\Lambda}^{\alpha_{\Lambda^c}}*D_{\beta',\Lambda}^{\alpha'_{\Lambda^c}} = D_{\beta+\beta',\Lambda}^{\alpha_{\Lambda^c}\circ \alpha'_{\Lambda^c}},
		\]
		where $\alpha''_{\Lambda^c}$ is the gluing of the paths $\alpha_{\Lambda^c}$ and $\alpha'_{\Lambda^c}$.
	\end{lemma}
	\begin{proof}
		Given an element $(\sigma_\Lambda, X)$ of the groupoid $\mathcal{G}_\Lambda$, the product is given by
		\begin{align*}	
			D_{\beta,\Lambda}^{\alpha_{\Lambda^c}}*D_{\beta',\Lambda}^{\alpha'_{\Lambda^c}}(\sigma_\Lambda, X) = \sum_{(\omega_\Lambda,Y)\in \mathcal{G}_\Lambda^{\iota_X\sigma_\Lambda}}D_{\beta,\Lambda}^{\alpha_{\Lambda^c}}(\omega_\Lambda, Y)D_{\beta',\Lambda}^{\alpha'_{\Lambda^c}}(\sigma_\Lambda,X\Delta Y).	
		\end{align*}
		For each $(\omega_\Lambda,Y)$, using the random representation of the densities, we have
		\begin{align*}
			D_{\beta,\Lambda}^{\alpha_{\Lambda^c}}(\omega_\Lambda,Y)&D_{\beta',\Lambda}^{\alpha'_{\Lambda^c}}(\sigma_\Lambda, X\Delta Y)= \\ &\left(\int_{\mathcal{P}_{\Lambda}^{\omega_\Lambda,Y}}e^{-\beta \int_{\alpha_\Lambda}H_\Lambda^{\alpha_{\Lambda^c}}(\phi)}S_\Lambda^{\alpha_{\Lambda_{\Lambda^c}}}(\alpha_\Lambda)d\nu_{\phi,\Lambda}(\alpha_\Lambda)\right)\left(\int_{\mathcal{P}_{\Lambda}^{\sigma_\Lambda,X\Delta Y}}e^{-\beta' \int_{\alpha_\Lambda}H_\Lambda^{\alpha'_{\Lambda^c}}(\phi)}S_\Lambda^{\alpha'_{\Lambda^c}}(\alpha_\Lambda)d\nu_{\phi,\Lambda}(\alpha_\Lambda)\right)
		\end{align*}
		Using Equation \eqref{equation_measure_definition} and making a change of variables we conclude the proof. 
	\end{proof}
	
	 When the path consists of only one jump at time $t$, then we get the following representation
	
	\[
	D_{\beta,\Lambda}^{\alpha_{\Lambda^c}} = e^{-t\beta H_\Lambda^\omega}*e^{-(1-t)\beta H_\Lambda^\eta}
	\]
	
	Thus, by iterating the above formula, we get this nice representation for the densities in terms of the usual Gibbs densities.
	\[
	D_{\beta,\Lambda}^{\alpha_{\Lambda^c}} = \prod_{j=1}^n e^{-s_j\beta H_\Lambda^{\omega_j}},
	\]
	where $\sum_{j=1}^n s_j = 1$. We readily see that it is not true that every operator $D_{\beta,\Lambda}^{\alpha_{\Lambda^c}}$ is self-adjoint (this happens, for example, when the path is symmetric, i.e., $\alpha_\Lambda = \xi_\Lambda \circ \xi^{-1}_\Lambda$, for some path $\xi_\Lambda$). 
	Using the auxiliary measure defined Equation \eqref{equation_measure_definition}, we can introduce the following definition for the finite volume Gibbs states. First, remember that for each $\Lambda$, we have the following random representation for the operator Given $(\omega_{\Lambda^c},X)\subset \mathcal{G}_{\Lambda^c}$, we can introduce the linear functional 
	
	\begin{align*}
	\mu_{\beta,\phi,\Lambda}^{\omega,X}(f) &= \frac{\mathbbm{1}(\omega_{\Lambda},X\cap \Lambda)}{Z_{\beta,\phi,\Lambda}^{\omega,X}}\sum_{\sigma_\Lambda \in \Omega_\Lambda} f * D_{\beta,\Lambda}(\sigma_\Lambda\omega_{\Lambda^c},X)\\
	&=\frac{\mathbbm{1}(\omega_{\Lambda},X\cap \Lambda)}{Z_{\beta,\phi,\Lambda}^{\omega,X}}\sum_{\substack{\sigma_\Lambda \in \Omega_\Lambda\\ (\eta_\Lambda,Y) \in \mathcal{G}_\Lambda^{\sigma_\Lambda}}} f(\eta_\Lambda\omega_{\Lambda^c},X\cup Y) D_{\beta,\Lambda}(\sigma_\Lambda\omega_{\Lambda^c},X\cup Y) \\
	&= \frac{\mathbbm{1}(\omega_{\Lambda},X\cap \Lambda)}{Z_{\beta,\phi,\Lambda}^{\omega,X}}\sum_{\substack{\sigma_\Lambda \in \Omega_\Lambda\\ (\eta_\Lambda,Y) \in \mathcal{G}_\Lambda^{\sigma_\Lambda}}} \int_{\mathcal{P}^{\sigma_\Lambda\omega_{\Lambda^c},X\cup Y}}f(\eta_\Lambda\omega_{\Lambda^c},X\cup Y) e^{-\beta \int_{\alpha_\Lambda}H_\Lambda^{\alpha_{\Lambda^c}}(\phi)}S_\Lambda^{\alpha_{\Lambda^c}}(\alpha_\Lambda)d\nu_{\phi}(\alpha),
	\end{align*}
	where we suppressed the $\mathbb{Z}^d$ in the notation, and $\mathbbm{1}$ is the identity operator. They satisfy the consistency condition
	
	\begin{proposition}
		$\Lambda' \subset \Lambda$. Then, for any $f \in C_c(\mathcal{G})$ we have
		\be
		\mu_{\beta,\phi,\Lambda}^{\omega,X}(f) = \mu_{\beta,\phi,\Lambda}^{\omega,X}(\mu_{\beta,\phi,\Lambda'}^{(\cdot)}(f))
		\ee
	\end{proposition}
	\begin{proof}
		Let us start by the definition of the integral in the region $\Lambda'$. We have that
		\begin{align*}
			\mu_{\beta,\phi,\Lambda'}^{\omega,X}(f) &= \frac{1}{Z_{\beta,\phi,\Lambda'}^{\omega,X}}\sum_{\sigma_{\Lambda'}\in\Omega_{\Lambda'}} f * D_{\beta,\Lambda'}(\sigma_{\Lambda'}\omega_{\Lambda'^c},X) \\
			&=\sum_{\sigma_{\Lambda'}\in\Omega_{\Lambda'}}\sum_{(\omega_{\Lambda'},Y)\in \mathcal{G}_{\Lambda'}^{\sigma_{\Lambda'}}} f(\omega,X\cup Y) 	G_{\Lambda'}(\sigma_{\Lambda'}\omega_{\Lambda'^c},X\cup Y),
		\end{align*}
		where
		\[
		G_{\Lambda'}(\sigma_{\Lambda'}\omega_{\Lambda'^c},X\cup Y) = \frac{D_{\beta,\Lambda'}(\sigma_{\Lambda'}\omega_{{\Lambda'}^c},X\cup Y)}{Z_{\beta,\phi,\Lambda'}^{\omega,X}}
		\]
		
		When one integrates with respect to the outside box $\Lambda$, one gets 
		
		\begin{align*}
			\mu_{\beta,\phi,\Lambda}^{\omega,X}(\mu_{\beta,\phi,\Lambda'}^{(\cdot)}(f)) &=  \sum_{\substack{\sigma_\Lambda\in\Omega_\Lambda\\(\omega_\Lambda,Y)\in \mathcal{G}_\Lambda^{\sigma_\Lambda}}} \mu_{\beta,\phi,\Lambda'}^{\omega,Y}(f) G_\Lambda(\sigma_\Lambda\omega_{\Lambda^c},X\cup Y) \\
			&=\sum_{\substack{\sigma_\Lambda\in\Omega_\Lambda\\(\omega_\Lambda,Y)\in \mathcal{G}_\Lambda^{\sigma_\Lambda}}}\left( \sum_{\substack{\eta_{\Lambda'}\in\Omega_{\Lambda'}\\(\tau_{\Lambda'},Y')\in \mathcal{G}_{\Lambda'}^{\eta_{\Lambda'}}}}f(\tau_{\Lambda'}\omega_{\Lambda^c},X\cup Y\Delta Y')G_{\Lambda'}(\eta_{\Lambda'}\sigma_{\Lambda^c},X\cup Y \Delta Y')\right) G_\Lambda(\sigma_\Lambda\omega_{\Lambda^c},X\cup Y)
		\end{align*}
		
		By breaking the sum and remembering the term that the identity $\mathbbm{1}$ is only different from zero in the unit space, we can break the sum 
		
		\[
		\sum_{\substack{\sigma_\Lambda\in\Omega_\Lambda\\(\omega_\Lambda,Y)\in \mathcal{G}_\Lambda^{\sigma_\Lambda}}} = \sum_{\substack{\sigma_{\Lambda\setminus\Lambda'}\in\Omega_{\Lambda\setminus\Lambda'}\\(\omega_{\Lambda\setminus\Lambda'},Y\setminus \Lambda')\in \mathcal{G}_{\Lambda\setminus \Lambda'}^{\sigma_{\Lambda\setminus \Lambda'}}}} \sum_{\sigma_{\Lambda'}\in\Omega_{\Lambda'}}
		\]
		And changing the order of the sum we get
		\[
		\mu_{\beta,\phi,\Lambda}^{\omega,X}(\mu_{\beta,\phi,\Lambda'}^{(\cdot)}(f)) =  \sum_{\substack{\sigma_{\Lambda^c}\in\Omega_{\Lambda^c}\\(\omega_{\Lambda\setminus\Lambda'},Y\setminus \Lambda')\in \mathcal{G}_{\Lambda\setminus \Lambda'}^{\sigma_{\Lambda\setminus \Lambda'}}}} F(\omega_{\Lambda'^c},X\cup Y\setminus \Lambda')
		\]
		
		where 
		
		\[
		F(\omega_{\Lambda'^c},X\cup Y\setminus \Lambda')= \sum_{\sigma_{\Lambda'}\in\Omega_{\Lambda'}}\left( \sum_{\substack{\eta_{\Lambda'}\in\Omega_{\Lambda'}\\(\tau_{\Lambda'},Y')\in \mathcal{G}_{\Lambda'}^{\eta_{\Lambda'}}}}f(\tau_{\Lambda'}\omega_{\Lambda^c},X\cup Y\Delta Y')G_{\Lambda'}(\eta_{\Lambda'}\sigma_{\Lambda^c},X\cup Y \Delta Y')\right)G_\Lambda(\sigma_\Lambda\omega_{\Lambda^c},X\cup Y)
		\]
		
		We can interchange the sums and make a trivial change of variables, yielding 
		\[
		F(\omega_{\Lambda'^c},X\cup Y\setminus \Lambda')= \sum_{\substack{\sigma_{\Lambda'}\in\Omega_{\Lambda'}\\(\omega_{\Lambda'},X\cap \Lambda')\in \mathcal{G}_{\Lambda'}^{\eta_{\Lambda'}}}}f(\omega_\Lambda,X\cup Y \Delta Y')G_{\Lambda'}(\sigma_\Lambda,X\cup Y \Delta Y')\left( \sum_{\eta_{\Lambda'}\in\Omega_{\Lambda'}}G_\Lambda(\eta_{\Lambda'}\sigma_{\Lambda\setminus\Lambda'}\omega_{\Lambda^c},X\cup Y)\right) 
		\]
		
		By the use of the consistency condition and Lemma 3.28 of \cite{LeNy} we get 
		
		\[
		G_{\Lambda'}(\sigma_\Lambda,X\cup Y \Delta Y')\left( \sum_{\eta_{\Lambda'}\in\Omega_{\Lambda'}}G_\Lambda(\eta_{\Lambda'}\sigma_{\Lambda\setminus\Lambda'}\omega_{\Lambda^c},X\cup Y)\right)   = G_\Lambda(\sigma_\Lambda,X\cup Y)
		\]
		
		Finishing the proof. 
		
	\end{proof}
	
	As one may notice, not every linear functional defined in this way can be a state. Indeed, since the identity operator $\mathbbm{1}$ is zero outside $\Omega$, we have $\mu^{\omega,X}_{\beta,\phi,\Lambda}$ every time $X\neq \emptyset$. But they are important to define proper maps in the $\text{C}^*$-algebra $C^*(\mathcal{G}_{\Lambda^c})$. Also, if we define the linear subspace
	\[
	V_{\omega,X} = (C^*(\mathcal{G}_\Lambda)\sigma_X^{(1)})\otimes C(\Omega_{\Lambda\setminus X}),
	\]
	for $X\subset \Lambda^c$. This is a linear subspace invariant by the adjoint operation, but is not an algebra for the usual product, since $(\sigma_X^{(1)})^2 = 1$. We proceed to show that the maps when $X=\emptyset$ are states. Some other properties are easily seen to be satisfied by these linear functionals, for instance, they are obviously continuous. 
	
	\begin{proposition}
		For every $\omega \in \Omega$, we have that $\mu_{\beta,\phi,\Lambda}^\omega$ is a state.
	\end{proposition}
\begin{proof}
We can use the fact the density of the Gibbs linear functional can be written as
\[
 \int_{\mathcal{P}^{\omega_{\Lambda^c}}_{\Lambda^c}} D_{\beta,\Lambda}^{\alpha_{\Lambda^c}}d\nu_{\phi,\Lambda^c}(\alpha_{\Lambda^c})
\]
 Notice that the adjoint operation, being antilinear, commutes with the integral, thus we have
 \[
\left(\int_{\mathcal{P}^{\omega_{\Lambda^c}}_{\Lambda^c}} D_{\beta,\Lambda}^{\alpha_{\Lambda^c}}d\nu_{\phi,\Lambda^c}(\alpha_{\Lambda^c})\right)^* = \int_{\mathcal{P}^{\omega_{\Lambda^c}}_{\Lambda^c}} D_{\beta,\Lambda}^{\alpha_{\Lambda^c}^{-1}}d\nu_{\phi,\Lambda^c}(\alpha_{\Lambda^c})
 \]
 Thus, since we are considering only paths that start and end in a configuration $\omega_{\Lambda^c}$, we have
\[
\left(\int_{\mathcal{P}^{\omega_{\Lambda^c}}_{\Lambda^c}} D_{\beta,\Lambda}^{\alpha_{\Lambda^c}}d\nu_{\phi,\Lambda^c}(\alpha_{\Lambda^c})\right)\left(\int_{\mathcal{P}^{\omega_{\Lambda^c}}_{\Lambda^c}} D_{\beta,\Lambda}^{\alpha_{\Lambda^c}}d\nu_{\phi,\Lambda^c}(\alpha_{\Lambda^c})\right)^* = \int_{\mathcal{P}^{\omega_{\Lambda^c}}_{\Lambda^c}} D_{2\beta,\Lambda}^{\alpha_{\Lambda^c}}d\nu_{\phi,\Lambda^c}(\alpha_{\Lambda^c}),
\]
by standard methods, as we applied in the proof of Lemma \ref{lemma_property_path}. The density is positive since an element of a $C^*$ algebra is positive if and only if is of the form $A^*A$, for some other element $A$. Since the linear functional is the trace against a positive operator, it must be a positive linear functional. It is normalized by definition.  
\end{proof}

The proof used, in a particular way, the fact that the beginning and the ending of the path are the same. If we tried to do the same with the diagonal measures, the integration with respect to the outside paths would change the start and the beginning. Actually, the best we can show is that the sum of the densities in the diagonal is a self-adjoint operator.
But more can be shown, by fixing the number of jumps that can occur at the boundary, a variant of the above lemma can be used to show that the density one gets is actually positive. Thus, let us introduce the following finite-volume Gibbs state
\[
\mu_{\beta,\phi,\Lambda}^{\omega,N}(f) = \frac{1}{Z_{\beta,\phi,\Lambda}^{\omega,N}}\sum_{\sigma_\Lambda \in \Omega_\Lambda} f*D_{\beta,\Lambda}^{\omega,N}(\sigma_\Lambda),
\]
where 
\[
D_{\beta,\Lambda}^{\omega,N} = \int_{\mathcal{P}^{\omega_{\Lambda^c},N}_{\Lambda^c}} D_{\beta,\Lambda}^{\alpha_{\Lambda^c}}d\nu_{\phi,\Lambda^c}(\alpha_{\Lambda^c}),
\]
is the density being integrated in the set of paths with exactly $n$ jumps, as 
\[
\mathcal{P}_{\Lambda^c}^{\omega_{\Lambda^c},N} \coloneqq \left\{\alpha_{\Lambda^c} = \sum_{i=1}^N\delta_{(\omega_{\Lambda^c,i},t_i,X_i)}: \omega_{\Lambda^c,1} = \omega_\Lambda, \omega_{\Lambda^c,N}= \iota_X\omega_\Lambda, \omega_{\Lambda^c, i-1} = \iota_{B_i}\omega_{\Lambda^c,i} , 1\leq i\leq N\right\},
\]
There is a special class of boundary conditions for the operators  that are related to paths where no arrival of operators happens at the boundary, i.e., the case where $N=0$ above. This means that $\alpha_{\Lambda^c}$ is constant through time. We will refer to these boundary conditions as \emph{classical boundary conditions}. These are important boundary conditions since they can be obtained directly by a Poisson point process representation of the Gibbs density of a specific Hamiltonian, which we will describe as follows. Let $\omega \in \Omega_{\Lambda^c}$ be a configuration and define the evaluation functional $\text{ev}_\omega$ on the dense subalgebra $C(\mathcal{G}_{\Lambda^c})$. By the definition of the norm in the regular representation, the evaluation functionals are actually states. Thus, we can form the conditional expectation $\text{Id} \otimes \text{ev}_{\omega_{\Lambda^c}}:C(\mathcal{G})\longrightarrow C(\mathcal{G}_\Lambda)$ and define the Hamiltonian with boundary condition $\omega$ by the expression $H^\omega_\Lambda(\phi) \coloneqq \text{Id}\otimes \text{ev}_{\sigma_{\Lambda^c}}(H_\Lambda(\phi) + W_\Lambda(\phi))$.
	
	There is no novelty in  construction above; it appeared before in Israel \cite{Is} as a proposal for boundary condition for quantum spin systems in much greater generality. Finally, we are motivated to introduce the following definition for the infinite volume Gibbs states
	\be
	\mathscr{G}_\beta(\phi) = \overline{\text{co}}\{\mu_\beta: \exists \{\Lambda_m\}_{m \geq 1} \text{   and   } \{\omega_m\}_{m\geq 1},\{N_m\}_{m\geq 1} \Lambda_m \nearrow \Z^d, \mu=w^*-\lim_{m\rightarrow \infty} \mu_{\beta,\phi,\Lambda_m}^{\omega_m,N_m}\}
	\ee

Notice that the finite volume Gibbs measures, when you fix a function $f$, is again a function of the boundary conditions itself. This motivates us to the following proposition. 

	\begin{proposition}
		For every $f \in C(\mathcal{G})$, we have that the following properties hold
  \begin{enumerate}
      \item If $f$ is self-adjoint, then so is $(\omega,X) \mapsto \mu_{\beta,\phi,\Lambda}^{\omega,X}(f)$.
      \item If $f$ is in $C(\mathcal{G}_{\Lambda_R^c})$, then
      $\mu_{\beta,\phi,\Lambda}(f) = f$.
  \end{enumerate}
	\end{proposition}
	\begin{proof}
		To prove the first assertion, just note that
  \begin{align*}
(\mu_{\beta,\phi,\Lambda}^{\omega,X}(f))^* &= \overline{\mu_{\beta,\phi,\Lambda}^{\iota_X\omega,X}(f)} \\
&= \frac{1}{\overline{Z_{\beta,\phi,\Lambda}^{\iota_X\omega,X}}}\sum_{\sigma_\Lambda \in \Omega_\Lambda} \overline{f*D_{\beta,\phi,\Lambda}}(\sigma_\Lambda(\iota_X\omega)_{\Lambda^c},X) \\
&=\frac{1}{Z_{\beta,\phi,\Lambda}^{\omega,X}}f*D_{\beta,\phi,\Lambda}(\sigma_\Lambda\omega_{\Lambda^c},X).
  \end{align*}
where the last equality is due the self-adjointness of $f$, the fact that
\[
(D_{\beta,\Lambda}^{\alpha_{\Lambda^c}})^* = D_{\beta,\Lambda}^{\alpha_{\Lambda^c}^{-1}},
\]
together with the bijection between paths between $\mathcal{P}_{\Lambda^c}^{\omega,X}$ and $\mathcal{P}_{\Lambda^c}^{\iota_X\omega,X}$
given by the involution $\alpha_{\Lambda^c}\mapsto \alpha_{\Lambda^c}^{-1}$. 

For the second point, notice that when a local function is in $C(\mathcal{G}_{\Lambda_R^c})$ there exists a $\Lambda' \subset \Lambda_R^c$ such that
\[
f(\sigma,X) = \mathbbm{1}_{\Lambda'^c}(\sigma_{\Lambda'^c},X\setminus \Lambda') f(\sigma_{\Lambda'},X\cap \Lambda')
\]
Then, one gets
\begin{align*}
f*D_{\beta,\Lambda}(\sigma_\Lambda\omega_{\Lambda},X) &= \sum_{(\eta_\Lambda,Y)\in \mathcal{G}_{\Lambda}^{\sigma_\Lambda}}f(\eta_\Lambda\omega_{\Lambda^c},X\cup Y)D_{\beta,\Lambda}(\sigma_\Lambda\omega_{\Lambda^c},X\cup Y) \\
&=f(\omega_{\Lambda'},X\cap \Lambda')\sum_{(\eta_\Lambda,Y)\in \mathcal{G}_{\Lambda}^{\sigma_\Lambda}}\mathbbm{1}_{\Lambda'^c}(\omega_{\Lambda'^c},(X\cup Y)\setminus \Lambda')D_{\beta,\Lambda}(\sigma_\Lambda\omega_{\Lambda^c},X\cup Y)
\end{align*}
and the last line is equal to $0$ if $(X\cup Y)\setminus\Lambda'\neq \emptyset$. Otherwise, there is only of $(\eta_\Lambda,Y)$ for which the sum is not zero, and this is when $Y=\emptyset$, thus we conclude the desired identity.   
  
\end{proof}

Actually we think the function $\mu_{\beta,\phi,\Lambda}(f)$ is positive whenever $f$ itself is positive. These properties allow us to introduce the following definition for a quantum specification on a groupoid.
\begin{definition}
    Let $\mathcal{G}$ be a groupoid with a decomposition $\mathcal{G}_\Lambda\times \mathcal{G}_{\Lambda^c}$ for every finite $\Lambda \subset \Z^d$. Then, a family of functions $\mu_\Lambda:C(\mathcal{G}_\Lambda)\times\mathcal{G}_{\Lambda^c}\rightarrow \mathbb{C}$ is called a proper quantum specification if and only if
    \begin{enumerate}
        \item For every $(\omega,X) \in \mathcal{G}_{\Lambda^c}$, $\mu_\Lambda^{\omega,X}$ is a linear functional; if $X=\emptyset$, then it is a state.
        \item For every $f \in C(\mathcal{G}_\Lambda)$, we know that $\mu_{\Lambda}(f)$ is a function in $C(\mathcal{G}_{\Lambda^c})$. More than that, if $f$ is self-adjoint, $\mu_\Lambda(f)$ is self-adjoint.
        \item There exists $\Lambda' \subset \Lambda$ such that if $f \in C(\mathcal{G}_{\Lambda'^c})$, then $\mu_\Lambda(f)=f$.
        \item For every $\Lambda'\subset \Lambda$, it holds $\mu_\Lambda(\mu_{\Lambda'}(f)) = \mu_\Lambda(f)$
    \end{enumerate}
\end{definition}

As we showed in Propositions 3.1, 3.2 and 3.3, the family of finite volume Gibbs functionals we introduce form a quantum specification for the groupoid in consideration. When this happens, we call this a \emph{quantum Gibbs specification}. This can be summarized in the following theorem.

	\begin{theorem}
		$\{\mu_{\beta,\Lambda}^{\bm{\omega}}\}_{\Lambda \in \mathcal{P}_f(\Z^d)}$ is a quantum specification. 
	\end{theorem}

 We are motivated to introduce the following definition for quantum DLR states
 
	\begin{definition}
		A state $\mu$ of $C^*(\mathcal{G})$ is said to be a quantum DLR state if it satisfies, for every $\Lambda$, 
		\[
		\mu_\beta(f)= \mu_\beta(\mu^{(\cdot)}_{\beta,\Lambda}(f)).
		\]
	\end{definition}
	The set of all $DLR$ states is
 \be
	\mathscr{G}_{\beta, DLR}(\phi) = \{\mu_\beta: \mu_\beta= \mu_\beta(\mu^{(\cdot)}_{\beta,\Lambda}), \forall \Lambda \in \mathcal{P}(\Z^d)\}
	\ee

 it is trivially a convex set. The following theorem holds.

\begin{theorem}
	$\mathscr{G}_\beta(\phi)=\mathscr{G}_{\beta, DLR}(\phi)$	
\end{theorem}
\begin{proof}
	The fact that $\mathscr{G}_\beta(\phi)\subset\mathscr{G}_{\beta, DLR}(\phi)$	follows by the consistency condition and the definition of $w^*$-convergence. For the other inclusion, suppose that there exists $\mu_\beta \in \mathscr{G}_{\beta,DLR}(\phi)\setminus\mathscr{G}_{\beta}(\phi)$. Since both sets are compact and convex, we know that there is a linear continuous functional $\varphi:C^*(\mathcal{G})\rightarrow \mathbb{C}$. and two real numbers $a,b$ such that
    \[
    \text{Re}(\varphi)(\mu_\beta) \leq a < b \leq \text{Re}(\varphi(\nu_\beta))
    \]
 for any $\nu_{\beta}\in \mathscr{G}_{\beta}(\phi)$, where $\text{Re}$ is the real part of the linear functional $\varphi$. Since every continuous linear functional in the $w^*$ topology is of the form $J_F$, for some $F \in C^*(\mathcal{G})$, we have that there exists a self-adjoint element of $C^*(\mathcal{G})$ where
  \[
    \mu_\beta(F) \leq a < b \leq \nu_\beta(F)
    \]
Since $F\in C^*(\mathcal{G})$, there exists a sequence of self-adjoint local functions converging to it. Thus, we can assume that there is a $\Lambda$ finite where $F\in C(\mathcal{G}_\Lambda)$. Using the DLR-equation, for a $\Lambda'$ containing $\Lambda$, we get that 
 \be\label{eq2}
\mu_{\beta,\phi,\Lambda}^{\sigma}(F) \leq a < b \leq \nu_\beta(F).
 \ee
Since the off-diagonal terms are zero, if $\Lambda'$ is large enough, we know that $\mu_{\beta,\phi,\Lambda}^{\sigma}(F)$ is a continuous function in $\Omega_{\Lambda'^c}$. Thus, by the Riesz-Markov theorem, the state $\mu_\beta$ restricts to a probability measure in $\Omega_{\Lambda'^c}$ and standard arguments allow us to conclude that 
 the Inequality \eqref{eq2} holds.
But since this holds for every $\Lambda'$ large enough, we can extract a sequence of finite volume Gibbs states that converges to some limit $\mu_\beta'$. But then $\mu_\beta' \in \mathcal{G}_\beta(\phi)$ and is separated by a linear functional. This yields a contradiction, therefore $\mathcal{G}_\beta(\phi)=\mathcal{G}_{\beta,DLR}(\phi)$.

\end{proof}

\section{The Relation between DLR and KMS states}

The local Hamiltonian operators can be used to define a local dynamics in $C^*(\mathcal{G}_\Lambda)$,
\[
\tau_t^\Lambda(A) = e^{-itH_\Lambda (\phi)}A e^{itH_\Lambda(\phi)}.
\]
 The finite volume Gibbs states have a nice algebraic relation with the dynamics at finite volume, called the KMS condition:
\[
\mu_{\beta,\Lambda}(AB) = \frac{\text{tr}(ABe^{-\beta H_\Lambda(\phi)})}{\text{tr}(e^{-\beta H_\Lambda(\phi)})} = \frac{\text{tr}(Ae^{-\beta H_\Lambda(\phi)}e^{\beta H_\Lambda(\phi)}Be^{-\beta H_\Lambda(\phi)})}{\text{tr}(e^{-\beta H_\Lambda(\phi)})} = \frac{\text{tr}(\tau_{i\beta}(B)Ae^{-\beta H_\Lambda(\phi)})}{\text{tr}(e^{-\beta H_\Lambda(\phi)})}= \mu_{\beta,\Lambda}(\tau_{i\beta}^\Lambda(B)A)
\] 

In a seminal work, Haag, Hugenholtz, and Winnink \cite{HHW} showed that the KMS condition survives the thermodynamic limit procedure, so it must encode a good definition of equilibrium state in infinite volume systems. 
\begin{definition}
Let $(\mathfrak{U},\tau)$ be a C*-dynamical system, i.e., a C*-algebra $\mathfrak{A}$ and a strongly continuous one-parameter group $\tau$. Let $\mu$ be a state. We say that this state satisfies the KMS condition if, and only if, for all $A \in \mathfrak{A}$ and analytic elements $B \in \mathfrak{A}$ the following holds:
\[
\mu(\tau_{i\beta}(B)A) = \mu(AB)
\] 
\end{definition}
Another definition of equilibrium states more familiar to physicists, known as the variational principle, can be seen in \cite{Bra2, Is, Simon}). The KMS condition cannot be used to directly define equilibrium states in an abelian C*-algebra because the dynamics is trivial (see Prop 5.3.28 in \cite{Bra2}). On the other hand, the Takesaki theorem (see \cite{Acc1, Acc2}) serves as a no-go theorem for the more straightforward generalizations of the DLR equations. Since we showed a suitable generalization to the quantum realm of the DLR equations in the previous section, we are led to the question of how these equations relate to the KMS states.
The first step for answering this question started with the classical case and was taken by Brascamp in \cite{Bras}, where he proved that when the interaction is classical and is embedded in a nonabelian C*-algebra, the KMS equation reduces to the DLR equations. 

 Sometime after that, Araki and Ion \cite{Ara2} closed the questions by defining a new condition for equilibrium, that we call here the \emph{Gibbs-Araki condition}, showed that when the interaction is classical the Gibbs-Araki condition reduces to the DLR equations and that the former is equivalent to the KMS condition in general. 
To define properly what the Gibbs-Araki Condition means we must introduce a new notion of perturbation of a state. We will explain this concept locally first. Let $\Lambda$ be a finite set and $\phi$ an interaction. Then we have the local dynamics $\tau^{\Lambda}_t$ defined in the beginning of this section and $P \in \mathfrak{U}_{\Lambda}$. For more details on the perturbation of the dynamics, see Chapter 4 of \cite{En} and \cite{Bra2}. Let $P \in C^*(\mathcal{G}_\Lambda)$ be a self-adjoint element and define the perturbation $\delta_P: \mathfrak{U}_\Lambda \rightarrow \mathfrak{U}_\Lambda$ by:
\[
 \delta_P(A) = i[H_\Lambda(\phi),A] +i [P,A]
\]
It is a standard result that this derivation generates a strongly continuous one-parameter group $\tau^P_t$ and it relates to $\tau_t$ by the Dyson series
\[
\tau_t^P(A) = \tau_t(A) + \sum_{n \geq 1} i^n \int_0^t dt_1 \int_0^{t_1} dt_2 ... \int_0^{t_{n-1}} dt_n [\tau_{t_n}(P),[...[\tau_{t_1}(P), \tau_t(A)]]].
\]
The above expansion is valid in much more general contexts, see \cite{Rob}. Furthermore, we can define a unitary operator
\[
\Gamma^P_t = \sum_{n \geq0} i^n \int_0^t dt_1 \int_0^{t_1}dt_2...\int_0^{t_{n-1}} dt_n \tau_{t_n}(P)...\tau_{t_1}(P).
\]
In \cite{Bra2} it is proved that this function is a cocycle and that the perturbed and the original dynamics are related by
\[
\tau^P_t = \Gamma_t^P \tau_t \Gamma_t^{P *}.
\] 
We are ready to introduce the perturbed state.
\begin{definition}
Let $\mathfrak{A}_n$ be the spin algebra, $\phi$ a short-range interaction, and $\tau$ the strongly continuous group generated by it. Let $P \in \mathfrak{A}_n$ be a self-adjoint element. Let $\mu$ be a state, then we define the perturbed state $\mu^P$ by
\[
\mu^P(A) = \frac{(\Omega_\mu^P,\pi_\mu(A) \Omega_\mu^P)}{(\Omega_\mu^P,\Omega_\mu^P)} 
\]
Where $\Omega^P_\mu = \pi_\mu(\Gamma_{\frac{i\beta}{2}}^P)\Omega_\mu$, $\Omega_\mu$ and $\pi_\mu$ are the cyclic vector and the representation of the GNS representation associated with $\mu$.
\end{definition}

\begin{definition}[Gibbs-Araki Condition]
Let $\mu:\mathfrak{A}_n\rightarrow \mathbb{C}$ be a state. We say that it satisfies the Gibbs-Araki condition for $\beta$ and interaction $\phi$ if, and only if
\begin{enumerate}
\item $\mu$ is faithful;
\item $\mu^P = \mu_\Lambda \otimes \bar{\mu}$, where $\mu^P$ is the perturbed state defined previously for $P = \beta W_\Lambda(\phi)$ and $\bar{\mu}$ is a state in $\mathfrak{A}_{\Lambda^c}$
\end{enumerate}  
\end{definition} 
One can find the proof of the next theorem in \cite{Bra2} in greater generality.
\begin{theorem}
Let $\mathfrak{A}$ be the spin algebra and $\phi$ a short-range interaction. Then the following assertions are equivalent for a state $\omega$
\begin{enumerate}
\item $\omega$ satisfies the Gibbs-Araki condition for $\phi$ at $\beta$;
\item $\omega$ is a $(\tau,\beta)-$ KMS state.
\end{enumerate}
\end{theorem}

 Define the function $E:\mathfrak{A}_n \rightarrow C(\Omega)$ on each local elements $f \in C(\mathcal{G})$, 
\[
 E(f)(\sigma,g)=\begin{cases}
f(\sigma) & g=e \\
0 & o.w.
 \end{cases}
\]
Clearly $\|E(f)\| \leq \|f\|$ and $E$ is linear, so there is a bounded extension to all $\mathfrak{A}_n$. 
\begin{definition}
Let $\mu$ be a state. We call it a classical state if and only if for all $A \in \mathfrak{A}_n$ 
\[
\mu(A) = \mu(E(A))
\]
\end{definition}
In other words, classical states can only see the classical part of the observables. Because $E(A)$ can be identified with a continuous function $f \in C(\Omega)$ and, by the Riesz-Markov theorem, $\omega$ can be identified with a probability measure $\mu$ on $(\Omega, \mathcal{F})$.

\begin{theorem}\label{t2}
Let $\mathfrak{A}_n$ be the spin algebra and $\phi$ a classical interaction. Then a state $\mu$ satisfies the Gibbs-Araki condition if, and only if, the state is classical and satisfies the DLR equations.
\end{theorem}
\begin{proof}
Assume, first, that our state satisfies the Gibbs-Araki condition. Let us show that if $E(A)=0$ then $\omega(A)=0$. Let $A$ be a local observable at some finite set $\Lambda$. We know that for the perturbation $P=\beta W_\Lambda$ the perturbed state satisfies:
\[
\omega^P(A) = \omega_\Lambda(A)
\]
Since the interaction is classical the operator $e^{-\beta H_\Lambda(\phi)}$ is diagonal and, if $E(A)=0$, then $E(A)e^{-\beta H_\Lambda(\phi)} = E(Ae^{-\beta H_\Lambda(\phi)}) =0$, so $\omega^P(A) = 0$. Again, since the interaction is classical, $e^{-\beta W_\Lambda} \in \mathcal{D}$, and, because $\mathcal{D} = \mathcal{D}_\Lambda \otimes \mathcal{D}_{\Lambda^c}$, we have an expansion for the surface energy term with relation to the elementary tensors:
\[
e^{\beta W_\Lambda} = \sum_{i,j \geq 1} A_i \otimes B_j
\]
Where $A_i \in \mathcal{D}_\Lambda$ and $B_j \in \mathcal{D}_{\Lambda^c}$. Calculating $\omega^P(Ae^{\beta W_\Lambda})$ we get:
\[
\omega^P(Ae^{\beta W_\Lambda}) = \sum_{i,j \geq 1} \omega_\Lambda(AA_i)\omega^P(B_j)
\]
Multiplying a matrix with zero diagonal by a diagonal matrix doesn't change its diagonal, so by our above reasoning $\omega_\Lambda(AA_i) = 0$ and, consequently, $\omega^P(Ae^{\beta W_\Lambda})=0 \Rightarrow \omega(A)=0$. 
This argument is independent of the initial set $\Lambda$, so for all $A \in \underset{\Lambda \in \mathcal{F}(\mathbb{Z}^d)}{\bigcup}\mathfrak{U}_\Lambda \cap ker E$. We claim that this is a dense subset of the kernel. Indeed, $ker E$ is complemented in $\mathfrak{U}$, so if we take $A_n \in \underset{\Lambda \in \mathcal{F}(\mathbb{Z}^d)}{\bigcup}\mathfrak{U}_\Lambda$ converging to $A \in ker(E)$ we can write:
\[
A_n = B_n + C_n, \;\; B_n \in ker E \;\; C_n \in Im E
\]
The projections are continuous, so we know that $B_n$ converges to $A$ and our assertion is correct. With this, we concluded that if $E(A)=0$ then $\omega(A)=0$. Now, every element $A \in \mathfrak{U}$ can be written as:
\[
A = B + C, \;\; C \in \mathcal{D}\;\; \text{and} \;\;\; E(B)=0
\]
By our previous considerations, it is clear that $\omega(A) = \omega(C) = \omega(E(A))$. For a fixed $\Lambda$, note that the Gibbs-Araki condition is equivalent to:
\begin{equation}\label{eq1}
\omega^P = \omega_\Lambda \otimes \omega^P = \omega^P(Id_{\Lambda^c} \otimes \omega_\Lambda)
\end{equation}
The Riesz-Markov theorem tells us that there are $\mu, \mu_P$ probability measures such that for all $f \in C(\Omega)$:
\[
\omega(f) = \int_\Omega f d\mu \;\;\; \omega^P(f) = \int_\Omega f d\mu_P
\]

Both measures are related by their Radon-Nykodim derivative:
\[
\frac{d \mu_P}{d\mu} = e^{\beta W_\Lambda}
\]

Equation \ref{eq1} tells us that the conditional expectation of $\mu_P$ relative to the $\sigma$-algebra $\mathcal{F}_{\Lambda^c}$ is, for all $f \in C(\Omega)$:
\[
\mathbb{E}_{\mu_P}(f)(\eta_{\Lambda^c}) = \omega_\Lambda(f)(\eta_{\Lambda^c}) = \frac{\sum_{\sigma_\Lambda \in \Omega_\Lambda}f(\sigma_\Lambda \eta_{\Lambda^c})e^{-\beta H_\Lambda(\phi)(\sigma_\Lambda)}}{\sum_{\sigma_\Lambda \in \Omega_\Lambda}e^{-\beta H_\Lambda(\phi)(\sigma_\Lambda)}}
\]

Denote $\left. \nu \right|_{\Lambda^c}$ the restriction of a measure $\nu$ to the $\sigma$ algebra $\mathcal{F}_{\Lambda^c}$. We want to calculate the conditional expectation for $\mu$. Note that:
\begin{equation}\label{eq2}
\int \mathbb{E}_\mu(f) d\left.\mu\right|_{\Lambda^c} = \int f d\mu = \int f \frac{d\mu}{d\mu_P} d\mu_P = \int \mathbb{E}_{\mu_P}\left(f \frac{d\mu}{d\mu_P}\right) d\left.\mu_P\right|_{\Lambda^c} 
\end{equation}
 When two measures are absolutely continuous with respect to each other then their restrictions are absolutely continuous too. The relation between the Radon-Nykodim derivatives is the following:
 \begin{equation}\label{eq3}
 \int \frac{d\left.\mu\right|_{\Lambda^c}}{d\left.\mu_P\right|_{\Lambda^c}}d\left.\mu_P\right|_{\Lambda^c} = \int d\left.\mu\right|_{\Lambda^c} = \int d\mu = \int \frac{d\mu}{d\mu_P}d\mu_P = \int \mathbb{E}_{\mu_P}\left(\frac{d\mu}{d\mu_P}\right)d\left.\mu_P\right|_{\Lambda^c}
 \end{equation}

The Equations \ref{eq2} and \ref{eq3} together give us:
\[
\mathbb{E}_\mu(f)\mathbb{E}_{\mu_P}\left(\frac{d\mu}{d\mu_P}\right)= \mathbb{E}_{\mu_P}\left(f\frac{d\mu}{d\mu_P}\right) 
\]
For all $f \in C(\Omega)$. Now, consider $\eta_{\Lambda^c} \in \Omega_{\Lambda^c}$ and let us calculate both sides of the equation above. The right hand side gives us:
\[
\mathbb{E}_{\mu_P}\left(f\frac{d\mu}{d\mu_P}\right)(\eta_{\Lambda^c}) = \frac{\sum_{\sigma_\Lambda \in \Omega_\Lambda}f(\sigma_\Lambda \eta_{\Lambda^c})e^{-\beta (H_\Lambda(\phi)(\sigma_\Lambda)+W_\Lambda(\sigma_\Lambda \eta_{\Lambda^c}))}}{\sum_{\sigma_\Lambda \in \Omega_\Lambda}e^{-\beta H_\Lambda(\phi)(\sigma_\Lambda)}}
\]
In the l.h.s we have
\[
\mathbb{E}_{\mu_P}\left(\frac{d\mu}{d\mu_P}\right)(\eta_{\Lambda^c}) = \frac{\sum_{\sigma_\Lambda \in \Omega_\Lambda} e^{-\beta (H_\Lambda(\phi)(\sigma_\Lambda)+W_\Lambda(\sigma_\Lambda \eta_{\Lambda^c}))}}{\sum_{\sigma_\Lambda \in \Omega_\Lambda}e^{-\beta H_\Lambda(\phi)(\sigma_\Lambda)}} 
\]
Since $H_\Lambda^\eta(\phi)(\sigma_{\Lambda}) = H_\Lambda(\phi)(\sigma_{\Lambda}) + W_\Lambda(\sigma_\Lambda \eta_{\Lambda^c})$ we conclude that:
\[
\mathbb{E}_\mu(f) = \mu_{\Lambda,\beta}^{\eta,\phi}(f)
\]

Now assume that the state is classical and satisfies the DLR equation. It is clear, by our previous calculations, that we can reverse the argument and conclude that all perturbation states $\omega^P$ are product states with the local Gibbs state and another state. Now we need to show that the state is faithful. Consider, first, $A \in \mathfrak{U}$ a positive element. We will show that if $\omega(A) = 0$ then $A=0$. 
It is a well-known result of measure theory that if $f$ is a positive function with $\int fd\mu =0$ then $f=0$ for $\mu$ a.e. If the support of the measure $\mu$ is the whole space, then the continuity of $f$ would imply that $f=0$. Indeed, if let $A \subset \Omega$ such that $f(x) = 0, \forall x \in A$. We know that $\mu(A) = 0$. We must have that $A^c$ in dense. Indeed, Suppose that $A^c$ is not dense i.e., $B= \overline{A^c}$ is a proper subset of $\Omega$. Because the measure is supported on the whole space, then $\mu(B)< \mu(X)$ and, then, $B^c$ must have a positive measure. But $B^c \subset A$, giving us a contradiction.
Now to show that $\mu$ is supported in the whole space, consider the cylinder sets $C_\Lambda^{j_1,...,j_{|\Lambda|}} = \{\sigma \in \Omega: \omega_{i_k} = j_k,\;\; 1 \leq k \leq |\Lambda|\}$. Because they are a basis to the topology of $\Omega$ if we can show that they have strictly positive measure then we are done.
\[
\int_\Omega \chi_{C_\Lambda^{j_1,...,j_{|\Lambda|}}} d\mu  =  \int_\Omega \mu_{\Lambda,\beta}^{\eta, \phi}(C_\Lambda^{j_1,...,j_{|\Lambda|}})d\mu  
\]

Because $\mu_{\Lambda,\beta}^{\eta, \phi}(C_\Lambda^{j_1,...,j_{|\Lambda|}})$ is a positive number that doesn't depend on the boundary condition $\eta$ we know that the cylinders have strictly positive measure. 

Now, consider $A$ a positive element of $\mathfrak{U}$ such that $\omega(A)=0$. By the above reasoning, $E(A)=0$. If $A$ is a local observable, then this implies that $E(A)=0$, and the only positive matrix with zero diagonal is the zero matrix. So, for all local observables $A$, if $\omega(A^*A)=0$ then $A=0$. 
\end{proof}
  
	\begin{theorem}
		$\mathscr{G}_{\beta, DLR}(\phi) \subseteq K_\beta(\phi)$.
	\end{theorem}
	\begin{proof}
		Let $\mu$ be a DLR state. That it is faithful is straightforward, since every positive operator can be approximated by local positive operators and, a simple application of the DLR property gives the result. We need to show that the perturbed state has the right property. Since the convergence of the dynamics is uniform in compacts, we can approximate the Dyson series in the thermodynamic limit by the local ones,
  \[
\Gamma_{i\beta}^{P} = \lim_{\Lambda'\nearrow \mathbb{Z}^d}\Gamma_{i\beta,\Lambda'}^{P},
  \]
  where $P = \beta W_\Lambda(\phi)$. By using the DLR equation for $\Lambda'_R$, we get
  \begin{align*}
\mu(f*\Gamma_{i\beta}^{P}) &= \lim_{\Lambda' \nearrow \mathbb{Z}^d}\mu(f*\Gamma_{i\beta,\Lambda'}^{P})\\
&= \lim_{\Lambda' \nearrow \mathbb{Z}^d}\mu(\mu_{\beta,\phi,\Lambda'_R}(f*\Gamma_{i\beta,\Lambda'}^{P})) \\
\lim_{\Lambda' \nearrow \mathbb{Z}^d}\mu(\mu_{\beta,\phi,\Lambda}\otimes\mu_{\beta,\phi,\Lambda'_R\setminus\Lambda}(f)) = \mu_{\beta,\phi,\Lambda}\otimes\widetilde{\mu}(f),
  \end{align*}
  where $\widetilde{\mu}$ is the composition of the right functionals. This shows that the perturbation is a product state with the empty boundary condition Gibbs state.
	\end{proof}
	
	\section{Final Remarks}
	
We expect that a proposal of boundary conditions in a suitable language can be helpful to study phase diagrams of quantum models.

	\section{Acknowledgements}
	
LA thanks FAPESP grants 2017/18152-2 and 2020/14563-0, the University of Victoria for the hospitality, where most of this work was done, and professor Walter Pedra for his guidance at the beginning of his career on the topic of $\text{C}^*$-algebras. RB was supported by CNPq grants 312294/2018-2 and 408851/2018-0, by FAPESP grant 16/25053-8, and by the University Center of Excellence \textquotedblleft Dynamics, Mathematical Analysis and Artificial Intelligence", at the Nicolaus Copernicus University, he thanks Roberto Fern\'{a}ndez for all the references and advice during his earlier years of career, in particular, the reference \cite{Gruber} and recently point out the boundary conditions on Simon's book \cite{Simon}. LA and RB thank Aernout van Enter for all the discussions and generosity over the years, for sharing your knowledge about mathematical physics and the literature of the area, and for his comments about the first versions of this manuscript. ML was supported by the Natural Sciences and Engineering Research Council of Canada, Discovery Grant RGPIN-2017-04052.

	\begin{appendices}

		\section{Basics of Point Process}
		
		Let $X$ be a Polish space and $\mathscr{B}(X)$ its Borel $\sigma$-algebra. The space of point measures on $X$, that we will denote by $\mathbb{N}(X)$ henceforth, is defined as
		\[
		\mathbb{N}(X)\coloneqq \{\mu = \sum_{i =1}^n \delta_{x_i}, x_i \in X, n \in \mathbb{N}\cup \{0, +\infty\}\}.
		\]
		Where the case $n=0$ is the null measure, i.e. $\mu_\emptyset(A)=0$ for every $A \in \mathscr{B}(X)$. Define for each measurable set $A\subset X$ the projection map $\pi_A : \mathbb{N}(X)\rightarrow \mathbb{N}\cup\{0,+\infty\}$ by
		\[
		\pi_A(\mu) = \mu(A),
		\]
		and consider $\mathscr{N}(X)$ be the smallest $\sigma$-algebra on $\mathbb{N}(X)$ such that all such projections are measurable. We are ready to define point processes
		\begin{definition}
			Let $(\Omega, \mathscr{A}, \mathbb{P})$ be a probability space. A point process is a measurable function $N: \Omega \rightarrow \mathbb{N}(X)$. 
		\end{definition}
		A point process is said to be \emph{simple} if for almost every $\omega \in \Omega$ it holds that $N(\omega)(\{x\})<2$, for any $x \in X$. Notice that, since the projections are measurable, we can define new measurable functions using the point process by
		\[
		N_A\coloneqq \pi_A \circ N.
		\]
		One can see $N(A)$ as a random choice of points inside $B$. 
		\begin{proposition}
			The following two assertions are equivalent
			\begin{enumerate}[label=(\roman*), series=l_after]
				\item $N:\Omega \rightarrow \mathbb{N}(X)$ is a point process;
				\item $N_A: \Omega \rightarrow \mathbb{N}\cup\{0,+\infty\}$ is measurable for every $A \in \mathscr{B}(X)$. 
			\end{enumerate}
		\end{proposition}
		\begin{proof}
			$(i)\Rightarrow (ii)$. Straightforward since $\pi_A$ and $N$ are measurable, the composition is measurable. 
			
			$(ii)\Rightarrow (i)$. For each $C \subset \mathbb{N}\cup\{0,+\infty\}$, we have that $\pi_B^{-1}(C)$ is measurable by definition of the $\sigma$-algebra $\mathscr{N}(X)$. To show that $N$ is a measurable function, it is sufficient to show that $N^{-1}(\pi_B^{-1}(C))$ is measurable. But $N^{-1}(\pi_B^{-1}(C)) = (\pi_B\circ N)^{-1}(C)$ thus it is measurable. 
		\end{proof}
		
		\begin{example}
			Let $(\Omega, \mathscr{A}, \mathbb{P})$ be a probability space and $\xi_i:\Omega \rightarrow X$, $i=1,\dots, n$ random variable. Then,
			\[
			N = \sum_{i=1}^n \delta_{\xi_i}.
			\]
		\end{example}
		
		A point process  $N:\Omega \rightarrow \mathbb{N}(X)$ is a \emph{Poisson Point Process} if for every $A \in \mathscr{B}(X)$ the two following conditions are satisfied
		\begin{itemize}
			\item $\mathbb{P}(N_A=k)= \frac{|\mu(A)|^k}{k!} e^{-\mu(A)}$, for any $k\geq 0$. 
			\item For any $B_1, \dots, B_m \in \mathscr{B}(X)$ pairwise disjoint the random variables $N(B_1), \dots, N(B_m)$ are independent.
		\end{itemize}
		
		We will show now that a Poisson point process exists and that it also has a representation as an empirical process. Let $\mu$ be a finite measure on $X$ and $N$ be the following point process
		\begin{equation}\label{Poisson_rep}
			N \coloneqq \sum_{i=1}^\tau \delta_{\xi_i},
		\end{equation}
		where $\tau, \xi_1,\xi_2,\dots$ are a countable family of independent random variables such that
		\begin{itemize}
			
			\item $\mathbb{P}(\tau = k) = \frac{\mu(X)^k}{k!}e^{-\mu(X)}$ 
			
			\item $\mathbb{P}(\xi_i \in B) = \frac{\mu(B)}{\mu(X)}$, for any $B\in \mathscr{B}(X)$ and $i=1,2, \dots$.
			
		\end{itemize}
		We will show first that such a family of random variables exists. Consider $\Omega = \mathbb{N}_0 \times \prod_{i\in \mathbb{N}} X$, with the product $\sigma-$algebra. We can define on it the product probability measure $\mathbb{P}$ on the cylinder sets by 
		\[
		\mathbb{P}(C \times \prod_{i \in \mathbb{N}} B_i) = P_\mu(X)(C)\times \prod_{i\in \mathbb{N}}\frac{\mu(B_i)}{\mu(X)},
		\]
		where $P_\mu(X)$ is the Poisson distribution with parameter $\mu(X)$. That this defines a probability measure on the product space $\Omega$ follows from \cite{Saeki}. Note that, by construction, the projections are independent random variables with the wanted distribution.
		
		\begin{proposition}
			Let $X$ be a Polish space with its $\mathscr{B}(X)$ Borel $\sigma$-algebra and $\mu$ a finite measure. Then the point process $N$ defined above is a Poisson Point Process. 
		\end{proposition}
		\begin{proof}
			Let $B \in \mathscr{B}(X)$. Using the independence of the random variables, we have
			\begin{align*}
				\mathbb{P}(N_B = k) &= \sum_{m\geq k}\mathbb{P}\left(\sum_{i=1}^m \delta_{\xi_i}(B)=k\right)\mathbb{P}(\tau = m) \\
				&= e^{-\mu(X)}\sum_{m\geq k}\frac{\mu(X)^m}{m!}\mathbb{P}\left(\sum_{i=1}^m \delta_{\xi_i}(B)=k\right)
			\end{align*}
			
			Hence
			\begin{align*}
				\mathbb{P}\left(\sum_{i=1}^m \delta_{\xi_i}(B)=k\right) &= \sum_{\substack{a_1+\dots+a_m = k \\ a_i=0,1}}\mathbb{P}\left(\delta_{\xi_i}(B) = a_i, i=1,\dots, m\right) \\
				&= \sum_{\substack{a_1+\dots+a_m = k \\ a_i=0,1}} \prod_{i=1}^m \mathbb{P}(\delta_{\xi_i}(B) = a_i),
			\end{align*} 
			where the last equality is due to the independence of the random variables. By our hypothesis on the random variables $\xi_i$, we have
			\[
			\mathbb{P}(\delta_{\xi_i}(B)= a_i) = \begin{cases}
				\frac{\mu(B^c)}{\mu(X)}, & a_i=0  \\
				\frac{\mu(B)}{\mu(X)}, & a_i = 1.
			\end{cases}
			\] 
			A standard stars and bars argument gives us 
			\[
			\sum_{\substack{a_1+\dots+a_m = k \\ a_i=0,1}} \prod_{i=1}^m \mathbb{P}(\delta_{\xi_i}(B) = a_i) = \binom{m}{k} \frac{\mu(B)^k\mu(B^c)^{m-k}}{\mu(X)^m}.
			\]
			Hence,
			\[
			\mathbb{P}(N_B = k) = \frac{\mu(B)^k}{k!}e^{-\mu(X)}\sum_{m\geq k} \frac{\mu(B^c)^{m-k}}{(m-k)!} = \frac{\mu(B)^k}{k!}e^{-\mu(B)}.
			\]
			
			Consider $B_1, B_2, \dots, B_m$ disjoint measurable sets. In order to show that the random variables $N_{B_1}, \dots, N_{B_m}$ are independent it is sufficient to show that
			\[
			\mathbb{P}(N_{B_i}= n_i, i=1, \dots, m) = \prod_{i=1}^m \mathbb{P}(N_{B_i} = n_i).
			\]
			We will assume that $\cup_{i=1}^m B_i = X$ now and show how to prove the general case later. With this assumption, necessarily we must have $\tau(\omega)= n= \sum_{i=1}^m n_i$. Thus, using independence we get
			\begin{align*}
				\mathbb{P}(N_{B_i}= n_i, i=1, \dots, m) &= \frac{\mu(X)^n}{n!}e^{-\mu(X)}\mathbb{P}\left(\sum_{i=1}^n\delta_{\xi_i}(B)= n_j, j=1, \dots, m\right) \\
				&=\sum_{\substack{a_{1,j}+\dots+a_{n,j} = n_j \\ a_{i,j}=0,1}}\mathbb{P}\left(\delta_{\xi_i}(B)= a_{i,j}, j=1, \dots, m, i=1,\dots, n\right)	
			\end{align*}
			Since the sets $B_1, \dots, B_m$ are disjoint, if $a_{i,j}=1$ for some $j$, then $a_{i,k}=0$ for any other $k\neq j$ since the opposite would imply that there is a point in $B_j \cap B_k$. Thus, using independence of the random variables $\xi_i$ we get
			\begin{align*}
				\sum_{\substack{a_{1,j}+\dots+a_{n,j} = n_j \\ a_{i,j}=0,1}}\mathbb{P}\left(\delta_{\xi_i}(B)= a_{i,j}, j=1, \dots, m, i=1,\dots, n\right) &= \sum_{\substack{a_{1,j}+\dots+a_{n,j} = n_j \\ a_{i,1}+\dots + a_{i,m}= 1 \\ a_{i,j}=0,1}}\prod_{i=1}^n\mathbb{P}\left(\delta_{\xi_i}(B)= a_{i,j}, j=1, \dots, m\right) \\
				& \sum_{\substack{a_{1,j}+\dots+a_{n,j} = n_j \\ a_{i,1}+\dots + a_{i,m}= 1 \\ a_{i,j}=0,1}}\prod_{i=1}^n\left(\frac{\mu(B_i)}{\mu(X)}\right)^{n_i}.
			\end{align*}
			Consider $a_{i,j}=0,1$ a solution to $a_{1,j}+\dots+a_{n,j} = n_j$, for any $j=1,\dots, m$ such that $a_{i,1}+\dots + a_{i,m}= 1$. The second equation says that for each $i$ there must be only one $j$ with $a_{i,j}\neq 0$. So we proceed in the following way. To produce a solution  to $a_{1,j}+\dots+a_{n,j} = n_j$ satisfying this constraint, we first choose $n_1$ indices $i$ to put as equal to $1$ and the rest we put equals to zero. For $j=2$ we now have $n-n_1$ indices avaible, so we choose $n_2$ of those to put as equal to $1$. We can proceed inductively until we reach the case $j=m$. This reasoning implies that the number of possible solutions $a_{i,j}$ is exactly
			\[
			\binom{n}{n_1}\binom{n-n_1}{n_2}\dots\binom{n-n_1-\dots-n_{m-1}}{n_m} = \frac{n!}{n_1!n_2!\dots n_m!}
			\]
			Hence,
			\[
			\mathbb{P}(N_{B_i}= n_i, i=1, \dots, m) = \frac{e^{-\mu(X)}}{n_1!n_2!\dots n_m!}\prod_{i=1}^m\mu(B_i)^{n_i},
			\]
			rearranging the terms and using that $\mu(X)=\mu(B_1)+\dots+\mu(B_m)$ yields the desired result. Consider now the general case, i.e., any family of disjoint measurable sets $B_1,\dots, B_m$. Write $B=\cup_{i=1}^m B_i$ and using our previous calculations we get
			\begin{align*}
				\mathbb{P}(N_{B_i}= n_i, i=1, \dots, m) &= \sum_{k\geq 0}\mathbb{P}(N_{B_i}= n_i, i=1, \dots, m, N_{X\setminus B}=k) \\
				&=\prod_{i=1}^m \mathbb{P}(N_{B_i} = n_i) \sum_{k\geq 1} \mathbb{P}(N_{X\setminus B} = k) = \prod_{i=1}^m \mathbb{P}(N_{B_i} = n_i)
			\end{align*}

		\end{proof}
		
		Poisson point process also has a uniqueness property in the sense that for any Poisson process with a given intensity measure $\mu$ are equal in distribution. The proof of this fact can be found in Theorem 1.2.1 in \cite{Reiss}. We need to introduce an important construction in measure theory before we discuss how to integrate functions with respect to a Poisson Point Process. 
		
		\begin{definition}
			The \textbf{coproduct} or the \textbf{disjoint union} of countably infinitely many measure spaces $(X_n,\mathscr{A}_n, \mu_n)$ is defined as
			\begin{align*}
				\bigsqcup_{n \in \mathbb{N}}X_n \coloneqq \bigcup_{n\in \mathbb{N}}\{(x,n): x \in X_n\}, \;\;\;\;
				\mathscr{A} \coloneqq \left\{\bigsqcup_{n\in \mathbb{N}} A_n : A_n \in \mathscr{A}_n\right\},
			\end{align*}
			and the measure of each set is given by
			\[
			\;\;\;\;
			\mu\left(\bigsqcup_{n\in \mathbb{N}} A_n\right) \coloneqq \sum_{n \in \mathbb{N}}\mu_n(A_n).
			\]
		\end{definition}
		
		It is easy to check that the set $\mathscr{A}$ is a $\sigma$-algebra and that $\mu$ is a measure. The fact that we have a countable collection of spaces in our definition does not mean that one needs to restrict to this case: we could, as well, consider an uncountable family of measurable spaces. Notice that we can also define injections $i_n: X_n \rightarrow \sqcup_{n \in \mathbb{N}} X_n$ by $i_n(x) = (x,n)$. These are measurable since 
		\[
		i_m^{-1}\left(\bigsqcup_{n \in \mathbb{N}}A_n\right) = A_m \in \mathscr{A}_m,
		\] 
		by definition of disjoint union. Also, it is easy to see that the image of measurable sets are measurable. Thus, if we have a family of measurable functions $f_n: X_n\rightarrow Y$, we can define a unique measurable function $f:\sqcup_{n \in \mathbb{N}}X_n\rightarrow Y$ by $f(i_n(x)) \coloneqq f_n(x)$. Indeed, given $B \in \mathscr{B}$ we have 
		\[
		f^{-1}(B) = \bigcup_{n \in \mathbb{N}}f^{-1}(B)\cap i_n(X_n) = \bigcup_{n \in \mathbb{N}}i_n(f_n^{-1}(B)),
		\]
		hence $f$ is measurable. This shows that the measurable functions on the coproduct as in one-to-one correspondence with sequences $\{f_n\}_{n\geq 1}$ of measurable functions $f_n:X_n\rightarrow Y$. 

		\begin{proposition}\label{Poisson_int}
			Let $f: \mathbb{N}(X)\rightarrow \mathbb{R}$ be a bounded measurable function and $N:\Omega \rightarrow \mathbb{X}$ a Poisson point process. Then, the following holds,
			\[
			\int_\Omega f\circ N(\omega)d\mathbb{P}(\omega) = f(\mu_\emptyset)e^{-\mu(X)}+e^{-\mu(X)}\sum_{n\geq 1}\frac{1}{n!}\int_{X^n} f(\delta_{x_1}+\dots+\delta_{x_n})d\mu^{\otimes n}(x_1,\dots,x_n),
			\]
			where $\mu^{\otimes n}$ is the $n$-fold product measure. 
		\end{proposition}
		\begin{proof}
			First, let $\tau, \xi_1,\xi_2, \dots$ be the random variables given by the Poisson point process \eqref{Poisson_rep}. Let $X^n$ be the product space, with the product $\sigma$-algebra. For the case $X^0 = \{0\}$ and $X^\infty$ the countable infinity product space with the cylinder $\sigma$-algebra. We will define the function  $\tilde{N}:\Omega\rightarrow \bigsqcup_{n \in \mathbb{N}\cup\{0,\infty\}}X^n$ given by
			\[
			\tilde{N}(\omega) = \begin{cases}
				(\xi_1(\omega),\xi_2(\omega), \dots ,\xi_{\tau(\omega)}(\omega)), & \tau(\omega) \neq 0 \\
				\mu_\emptyset& \tau(\omega)=0
			\end{cases} 
			\]
			This function is measurable. Given an integrable function $f:\bigsqcup_{n \in \mathbb{N}\cup\{0,\infty\}}X^n \rightarrow \mathbb{R}$ and $f_n$ its restrictions to the subspace $X^n$. We have
			\[
			\int_\Omega f \circ \tilde{N} d\mathbb{P}(\omega) = f_0(0)e^{-\mu(X)}+\sum_{n\geq 1} \int_{\tau = n} f_n(\xi_1(\omega), \dots, \xi_n(\omega)) d\mathbb{P}(\omega)
			\]
			Suppose that $f_n = \mathbbm{1}_{B_1 \times \dots \times B_n}$, for measurable sets $B_i \in \mathscr{B}(X)$. Independence of the random variables $\xi_i$ yields
			\begin{align*}
				\int_{\tau = n} f_n(\xi_1(\omega), \dots, \xi_n(\omega)) d\mathbb{P}(\omega) &= \mathbb{P}(\tau=n)\prod_{i=1}^n \mathbb{P}(\xi_i \in B_i) \\
				&= \frac{e^{-\mu(X)}}{n!}\int_{X^n} \mathbbm{1}_{B_1 \times \dots \times B_n}(x_1,\dots,x_n) d\mu^{\otimes n}(x_1,\dots,x_n).
			\end{align*}
			Standard measure theoretic techniques allow us to extend the above result for general integral functions. Hence
			\[
			\int_\Omega f \circ \tilde{N} d\mathbb{P}(\omega) = f_0(0)e^{-\mu(X)}+\sum_{n\geq 1}\frac{e^{-\mu(X)}}{n!}\int_{X^n}f_n(x_1,\dots,x_n) d\mu^{\otimes n}(x_1,\dots,x_n).
			\]
			Let the function $\varphi:\mathbb{N}(X)\rightarrow \bigsqcup_{n \in \mathbb{N}\cup\{0,\infty\}}X^n$ be defined by
			\[
			\varphi_n(x_1,\dots, x_n)= \sum_{i=1}^n \delta_{x_i}.
			\] 
			Notice that $N = \varphi \circ \tilde{N}$. We claim that this function is measurable. Indeed, by previous considerations, we only need to show that $\varphi_{n,A}:X^n \rightarrow \mathbb{N}_0$ is measurable. It is sufficient to show that the pre-images of the singletons are measurable. Thus
			\[
			\varphi^{-1}_{n,A}(\{m\}) = \begin{cases}
				\emptyset & n<m, \\
				\bigcup_{\sigma \in F} B_\sigma & n\geq m,
			\end{cases}
			\]
			where $F =\{\sigma:\{1,\dots,n\}\rightarrow \{0,1\}: \sum_{i=1}^n \sigma(i) = m\}$, $B_\sigma = \prod_{i=1}^n A_{\sigma(i)}$, where $A_1= A$ and $A_0 = A^c$. Hence for any $f:\mathbb{N}(X)\rightarrow \mathbb{R}$, it holds
			\begin{align*}
				\int_\Omega f \circ N(\omega) d\mathbb{P}(\omega) &=\int_\Omega f \circ(\varphi \circ \tilde{N})(\omega) d\mathbb{P}(\omega) \\
				& = f_0(\mu_\emptyset)e^{-\mu(X)}+e^{-\mu(X)}\sum_{n\geq 1}\frac{1}{n!}\int_{X^n}f(\delta_{x_1}+ \dots + \delta_{x_n}) d\mu^{\otimes n}(x_1,\dots,x_n).
			\end{align*}
		\end{proof}
		
		Another important example that we will use to construct random representations for spin systems is the Bernoulli point process. We will focus on a more concrete case, where $X=[0,1]$. Given two point process $N, \tilde{N}$ it is straightforward to see that $N+\tilde{N}$ is again a point process. 
		Consider $\lambda \in \mathbb{R}$ and $\{\xi_{i,j}\}_{i\in \mathbb{N}, j =1,\dots, n}$ a sequence of i.i.d variables such that
		\[
		\mathbbm{P}(\xi_{n,j} = 0) = 1 - \mathbbm{P}(\xi_{n,j} = 1) = \frac{\lambda}{n},
		\]
		for $1\leq j \leq n$. These are probabilities for $n$ large enough. Define the point process
		\be\label{BPP}
		N_n(x,B) = \sum_{j=1}^n \xi_{n,j}(x)\delta_{\frac{j}{n}}(B),
		\ee
		where $B\subset [0,1]$ is a Borel set. 
		
		\begin{proposition}\label{integformpoirep:app1}
			Let $N_n$ be the Bernoulli point process defined in \eqref{BPP}. Then, we have that 
			\[
			\int_\Omega f \circ N_n (\omega) d\mathbb{P}(\omega) = \sum_{m\geq 0} \sum_{\substack{j_l \in \{0,\dots,n-1\} \\ 1 \leq l \leq k}} f\left(\delta_{\frac{j_1}{n}}+ \dots +\delta_{\frac{j_m}{n}}\right)\left(1- \frac{\lambda}{n}\right)^{n-m}\left(\frac{\lambda}{n}\right)^m 
			\]
		\end{proposition}
		
		\begin{proof}
			The strategy of this proof will be the same as the one employed in Proposition \ref{Poisson_int}. Let $\tilde{N}: \Omega \rightarrow \bigsqcup_{n \in \mathbb{N}_0}[0,1]^n$ defined by
			\[
			\tilde{N}_n(\omega) = \begin{cases}
				(\frac{j_1}{n}, \dots, \frac{j_k}{n}), & \xi_{n,j_l}=1 \;\; \mathrm{ and } \;\;\xi_{n,j}=0, j\neq j_l, 1\leq l \leq k, \\
				0, & \sum_{j=1}^n \xi_{n,j}=0.
			\end{cases}
			\]
			It holds 
			\[
			\int_\Omega f\circ \tilde{N}_n(\omega) d\mathbb{P}(\omega) = \sum_{m\geq 0} \int_{\sum_j \xi_{n,j}=m} f_m \circ \tilde{N}_n(\omega)d\mathbb{P}(\omega). 
			\]
			It is straightforward to see that the r.h.s of the equation above is equal to
			\[
			\sum_{\substack{j_l \in \{0,\dots,n-1\}\\ 1 \leq l \leq k}}f\left(\frac{j_1}{n}, \dots, \frac{j_k}{n}\right) \left(1-\frac{\lambda}{n}\right)^{n-m}\left(\frac{\lambda}{n}\right)^m.
			\]
			Using the map $\varphi$ defined in Proposition \ref{Poisson_int} yields the desired result. 
		\end{proof}
		
		\begin{corollary}\label{corol1:appA}
			Let $N_{n,l}$, $l=1,..., m$, be Bernoulli point processes with probability $\lambda_l/n$. Then, for $N_n = \sum_l N_{n,l}$
			\[
			\int_\Omega f \circ N_n(\omega) d\mathbb{P}(\omega) = 
			\]
		\end{corollary}
		\begin{proof}
			\[
			\int_\Omega f \circ N_n(\omega) d\mathbb{P}(\omega) = \sum_{M\geq 0}\sum_{M_1+\dots+M_m = M} \int_{\sum_k\xi_{n,k,l}=M_l} f \circ N_n(\omega) d\mathbb{P}(\omega)
			\]
			
			\[
			\int_{\sum_k\xi_{n,k,l}=M_l} f \circ N_n(\omega) d\mathbb{P}(\omega)=\sum_{\substack{J_l \in \mathcal{P}(\{0,\dots,n-1\}) \\ 1 \leq l \leq m}}f\left(\sum_{l=1}^m\sum_{j\in J_l} \delta_{\frac{j}{n}}\right)\prod_{l=1}^m\left(1-\frac{\lambda_l}{n}\right)^{n-M_l}\left(\frac{\lambda_l}{n}\right)^{M_l}
			\]
		\end{proof}
		
		Suppose that we have a sequence of probability measures $\mu_n$ in the coproduct space $\bigsqcup_{m \in \mathbb{N}_0}[0,1]^m$. Then, one can define measures on $[0,1]^m$ by restriction. Let these restrictions be denoted by $\mu_{n,m}$. Then, if we have that each $\mu_{n,m}$ converges weakly to a $\mu_m$, then the monotone convergence theorem implies that $\mu_n$ converges to $\mu = \sum_m \mu_m$. Let $B_1,\dots, B_m$ be Borel sets in $[0,1]$. Then, for a continuous function $f:[0,1]^m \rightarrow \mathbb{R} $, we have 
		\[
		\sum_{\substack{j_l \in \{0,\dots,n-1\}\\ 1 \leq l \leq k}}f\left(\frac{j_1}{n}, \dots, \frac{j_k}{n}\right) \left(1-\frac{\lambda}{n}\right)^{n-m}\left(\frac{\lambda}{n}\right)^m =\left(1-\frac{\lambda}{n}\right)^{n-m}\int_{[0,1]^m} g_n(x) d\lambda^{\otimes n},
		\]
		where the function $g_n:[0,1]^m \rightarrow \mathbb{R}$ is defined by 
		\[
		g_n(x) = f\left(\frac{j_1}{n}, \dots, \frac{j_l}{n}\right), \quad \textrm{if} \;\; x_l \in \left(\frac{j_l-1}{n}, \frac{j_l}{n}\right],
		\]
		and $d\lambda = \lambda dx$, where $dx$ is the Lebesgue measure. Notice that $\lim_{n\rightarrow \infty} g_n = f$ pointwise. The Lebesgue dominated convergence theorem gives us that,
		\[
		\lim_{n\rightarrow \infty} \left(1-\frac{\lambda}{n}\right)^{n-m}\int_{[0,1]^m}g_n(x) d\lambda^{\otimes m} = e^{-\lambda}\int_{[0,1]^m}f(x)d\lambda^{\otimes m}.
		\]
		Thus, we get that the Bernoulli point processes converge weakly to a Poisson point process with intensity measure $d\lambda$. In the case both are independent Poisson Point processes, the sum is again a Poisson point process, as the following proposition shows.
		\begin{proposition}
			Let $N,\tilde{N}:\Omega \rightarrow \mathbb{N}(X)$ be two Poisson point processes. Then the point process $N+\tilde{N}$ is Poisson.
		\end{proposition}
		\begin{proof}
			Take $B\in \mathscr{B}(X)$, and consider $N_B+\tilde{N}_{B}$. Then, the independence of $N$ and $\tilde{N}$ imply the independence of $N_B$ and $N_B$, thus
			\begin{align*}
				\mathbb{P}(N_B+ \tilde{N}_{B}=k)&=\sum_{j=0}^k \mathbb{P}(N_B=k-j)\mathbb{P}(\tilde{N}_{B}=j) \\
				& = \sum_{j=0}^k \frac{\mu(B)^{k-j}\nu(B)^j}{(k-j)!j!}e^{-(\mu(B)+\nu(B))} \\
				& = \frac{e^{-(\mu(B)+\nu(B))}}{k!}\sum_{j=0}^k \binom{k}{j}\mu(B)^{k-j}\nu(B)^j = \frac{(\mu(B)+\nu(B))^ke^{-(\mu(B)+\nu(B))}}{k!}. 
			\end{align*}
			
			Consider $B_1, \dots, B_m$ disjoint measurable sets. 
			\[
			\mathbb{P}(N_{B_i}+ \tilde{N}_{B_i}=k_i, i=1,\dots,m) = \sum_{\substack{k_{i,1}+k_{i,2}=k_i \\ i=1, \dots, m}} \mathbb{P}(N_{B_i}=k_{i,1}, \tilde{N}_{B_i}=k_{i,2}, i=1, \dots, m)  
			\]
			
			Since $N$ and $\tilde{N}$ are independent and the following holds, we have
			\begin{align*}
				\mathbb{P}(N_{B_i}=k_{i,1}, \tilde{N}_{B_i}=k_{i,2}, i=1, \dots, m)
				&=  \mathbb{P}(N \in \pi_{B_i}^{-1}(\{k_{i,1}\}), \tilde{N} \in \pi_{B_i}^{-1}(\{k_{i,2}\}) i=1, \dots, m) \\
				&=  \mathbb{P}(N \in \pi_{B_i}^{-1}(\{k_{i,1}\})i=1, \dots, m)  \mathbb{P}(\tilde{N} \in \pi_{B_i}^{-1}(\{k_{i,2}\}) i=1, \dots, m).
			\end{align*}
			Hence,
			\begin{align*}
				\mathbb{P}(N \in \pi_{B_i}^{-1}(\{k_{i,1}\}), \tilde{N} \in \pi_{B_i}^{-1}(\{k_{i,2}\}) i=1, \dots, m) &=\prod_{i=1}^m\left(\sum_{k_{i,1}+k_{i,2}=k_i} \mathbb{P}(N \in \pi_{B_i}^{-1}(\{k_{i,1}\}))  \mathbb{P}(\tilde{N} \in \pi_{B_i}^{-1}(\{k_{i,2}\}))\right) \\
				&=\prod_{i=1}^m\sum_{k_{i,1}+k_{i,2}=k_i}\mathbb{P}(N_{B_i}=k_{i,1}, \tilde{N}_{B_i}=k_{i,2}).
			\end{align*}
			This yields the desired result. 
		\end{proof}
		
		The probability that two independent poisson point processes with nonatomic probability measures $\mu, \nu$ will draw the same point is zero. Indeed, let $x\in X$ be a point and $B_{\frac{1}{n}}(x)$ be the balls centered at $x$ with radius $1/n$. Then we have that
		\[
		\mathbb{P}(N_{B_{\frac{1}{n}}(x)},\tilde{N}_{B_{\frac{1}{n}}(x)}\geq 1) = \mathbb{P}(N_{B_{\frac{1}{n}}(x)}\geq 1)\mathbb{P}(\tilde{N}_{B_{\frac{1}{n}}(x)}\geq 1)= (1-e^{-\nu(B_{\frac{1}{n}}(x))})(1-e^{-\mu(B_{\frac{1}{n}}(x))}).
		\]
		Since the measure $\nu$ is finite, we have that $\lim_{n\rightarrow\infty}\nu(B_{\frac{1}{n}}(x)) = \nu(\{x\})=0$. The same holds for $\mu$. Thus we have that
		\[
		\mathbb{P}(N,\tilde{N} \text{  draws the point   }x)=\lim_{n\rightarrow\infty} \mathbb{P}(N_{B_{\frac{1}{n}}(x)},\tilde{N}_{B_{\frac{1}{n}}(x)}\geq 1) = 0.
		\]
		The result above implies that when we have a finite number of independent Poisson point process we can associate to each draw a definite label allowing us to integrating more general functions, that even have a dependence on these labels. 
		\begin{corollary}\label{Corol_ppp}
			Let $N_i$, for $i=1,\dots,M$ be independent Poisson point processes on $[0,1]$ with intensity measures $\lambda_i dt$, for $\lambda_i>0$. Let $N = \sum_{i} N_i$ and $f:\mathbb{N}([0,1]\times\{1,\dots,M\})\rightarrow \mathbb{R}$ be a bounded measurable function. Then, it holds
			\[
			\int f\circ N (\omega) d\mathbb{P}(\omega) = e^{-\beta_i \sum  \lambda_i}\sum_{n\geq 0} \frac{1}{n!}\int_{[0,1]^n}  \sum_{i_1,\dots,i_n\in [i]} f\left(\sum_{j=1}^n\delta_{t_j,i_j}\right) \prod_{j=1}^n \lambda_j dt^n 
			\]
		\end{corollary}
		
	\end{appendices} 
	
\end{document}